\documentclass[1p]{elsarticle}

\journal{TCS}
\usepackage{amsfonts,epsfig,amsmath, amssymb, tabularx}
\usepackage{algorithm}
\usepackage{algpseudocode}
\usepackage{framed}
\usepackage{tabularx}
\usepackage{textgreek}
\usepackage{stmaryrd, subcaption, graphicx}

\newtheorem{theorem}{Theorem}[section]

\newtheorem{corollary}[theorem]{Corollary}
\newtheorem{lemma}[theorem]{Lemma}
\newtheorem{definition}[theorem]{Definition}
\newtheorem{property}[theorem]{Property}

\newtheorem{observation}[theorem]{Observation}
\newenvironment{proof}{\noindent{\bf Proof~}}{\null\hfill $\Box$\par\medskip}

\newcolumntype{L}{>{\raggedright\arraybackslash}X}
\newcolumntype{R}{>{\raggedleft\arraybackslash}X}
\newcolumntype{C}{>{\centering\arraybackslash}X}

\newcommand{\problemDecision}[3]
{
 \begin{framed}
   \centerline{{\sc #1}}     
   \bigskip
   \noindent
   \begin{tabularx}{\textwidth}{lX}
     \emph{Given:} & #2	\\
     \emph{Question:} & #3
   \end{tabularx}
 \end{framed}
}

\newcommand{\mdSliceBC}{
  $k$-bounded BC Metric Dimension
}

\pdfoutput=1

\begin{document}

\begin{frontmatter}
\title{Computing the metric dimension by decomposing graphs into extended biconnected components}
\author[Duesseldorf]{Duygu Vietz} 
\ead{Duygu.Vietz@hhu.de}
\author[Duesseldorf]{Stefan Hoffmann}
\ead{Stefan.Hoffmann@hhu.de}
\author[Duesseldorf]{Egon Wanke} 
\ead{E.Wanke@hhu.de}
\address[Duesseldorf]{Heinrich-Heine-Universit\"at D\"usseldorf, Institute of Computer Science, 40225 D\"usseldorf, Germany}

\begin{abstract}
A vertex set $U \subseteq V$ of an undirected graph $G=(V,E)$ is a \textit{resolving set} for $G$, if for every two distinct vertices $u,v \in V$ there is a vertex $w \in U$ such that the distances between $u$ and $w$ and the distance between $v$ and $w$ are different. The \textit{Metric Dimension} of $G$ is the size of a smallest resolving set for $G$. Deciding whether a given graph $G$ has Metric Dimension at most $k$ for some integer $k$ is well-known to be NP-complete. Many research has been done to understand the complexity of this problem on restricted graph classes. In this paper, we decompose a graph into its so called \textit{extended biconnected components} and present an efficient algorithm for computing the metric dimension for a class of graphs having a minimum resolving set with a bounded number of vertices in every extended biconnected component. Further we show that the decision problem {\sc Metric Dimension} remains NP-complete when the above limitation is extended to usual biconnected components.
\end{abstract}

\end{frontmatter}


\section{Introduction}
\label{secIntro}

An undirected graph $G=(V,E)$ has metric dimension at most $k$ if there is a vertex set $U \subseteq V$ such that $|U| \leq k$ and $\forall u,v \in V$, $u \not=v$, there is a vertex $w \in U$ such that $d_G(w,u) \not= d_G(w,v)$, where $d_G(u,v)$ is the distance (the length of a shortest path in an unweighted graph) between $u$ and $v$. The metric dimension of $G$ is the smallest integer $k$ such that $G$ has metric dimension at most $k$. The metric dimension was independently introduced by Harary and Melter \cite{HM76} and Slater \cite{Sla75}.

If for three vertices $u,v,w$, we have $d_G(w,u) \not= d_G(w,v)$, then we say that $u$ and $v$ are {\em resolved} by vertex $w$. If every pair of vertices is resolved by at least one vertex of a vertex set $U$, then $U$ is a {\em resolving set} for $G$. The {\em metric dimension} of $G$ is the size of a minimum resolving set. Such a smallest resolving set is also called a {\em resolving basis} for $G$. In certain applications, the vertices of a resolving set are also called {\em resolving vertices}, {\em landmark nodes} or {\em anchor nodes}. This is a common naming particularly in the theory of sensor networks. 

Determining the metric dimension of a graph is a problem that has an impact on multiple research fields such as chemistry \cite{CEJO00}, robotics \cite{KRR96}, combinatorial optimization \cite{ST04} and sensor networks \cite{HW12}. Deciding whether a given graph $G$ has metric dimension at most $k$ for a given integer $k$ is known to be NP-complete for general graphs \cite{GJ79}, planar graphs \cite{DPSL12}, even for those with maximum degree 6 and Gabriel unit disk graphs \cite{HW12}. Epstein et al.\ showed the NP-completeness for split graphs, bipartite graphs, co-bipartite graphs and line graphs of bipartite graphs \cite{epstein2015weighted} and Foucaud et al.\ for permutation and interval graphs \cite{foucaud2015algorithms}\cite{foucaud2017identification}.

There are several algorithms for computing the metric dimension in polynomial time for special classes of graphs, as for example for trees \cite{CEJO00,KRR96}, wheels \cite{HMPSCP05}, grid graphs \cite{MT84}, $k$-regular bipartite graphs \cite{SBSSB11}, amalgamation of cycles \cite{IBSS10}, outerplanar graphs \cite{DPSL12}, cactus block graphs \cite{hoffmann2016linear} and chain graphs \cite{fernau2015computing}. The approximability of the metric dimension has been studied for bounded degree, dense, and general graphs in \cite{HSV12}. Upper and lower bounds on the metric dimension are considered in \cite{CGH08,CPZ00} for further classes of graphs.

There are many variants of the Metric Dimension problem. For the weighted version Epstein et al.\ gave a polynomial-time algorithm on paths, trees and cographs \cite{epstein2015weighted}. Hernando et al.\ investigated the fault-tolerant Metric Dimension in \cite{hernando2008fault}, Estrada-Moreno et al.\ the $k$-metric Dimension in \cite{estrada2013k} and Oellermann et al.\ the strong metric Dimension in \cite{oellermann2007strong}.
 
The parameterized complexity was investigated by Hartung and Nichterlein. They showed that for the standard parameter the problem is $W[2]$-complete on general graphs, even for those with maximum degree at most three \cite{hartung2013parameterized}. Foucaud et al.\ showed that for interval graphs the problem is FPT for the standard parameter \cite{foucaud2015algorithms}\cite{foucaud2017identification}. Afterwards Belmonte et al.\ extended this result to the class of graphs with bounded treelength, which is a superclass of interval graphs and also includes chordal, permutation and AT-free graphs \cite{belmonte2017metric}. 
  
\bigskip 
In this paper, we introduce a concept that allows us to compute the metric dimension based on a tree structure given by the decomposition of a graph $G$ into components like {\em bridges}, {\em legs}, and so-called {\em extended biconnected components}. An {\em extended biconnected component} $H$ of $G$ is an induced subgraph of $G$ formed by a biconnected component $H'$ of $G$ extended by paths attached to vertices of the biconnected component $H'$. Each vertex of $H'$ has at most one path attached to it. Each vertex at which a path is attached is a separation vertex in $G$ and not adjacent to any vertex outside of extended biconnected component $H$. The idea of such a decomposition leads to a polynomial time solution for the Metric Dimension problem restricted to graphs having a minimum resolving set with a bounded number of vertices in every extended biconnected component. This result is especially noteworthy, because we also show that the decision problem {\sc Metric Dimension} remains NP-complete if the above limitation is extended to usual biconnected components.

\section{Definitions and Basic Terminology}

We consider {\em graphs} $G=(V,E)$, where $V$ is the set of {\em vertices} and $E$ is the set of {\em edges}. We distinguish between {\em undirected graphs} with edge sets $E \subseteq \{\{u,v\}~|~u,v \in V,~u \not=v\}$ and {\em directed graphs} with edge sets $E \subseteq V \times V.$ Graph $G'=(V',E')$ is a {\em subgraph} of $G=(V,E)$ if $V' \subseteq V$ and $E' \subseteq E$. It is an {\em induced subgraph} of $G$, denoted by $G|_{V'}$, if $E' = E \cap \{\{u,v\}~|~ u,v \in V'\}$ or $E' = E \cap (V'\times V')$, respectively.

A sequence of $k+1$ vertices $(u_1,\ldots,u_{k+1})$, $k \geq 0$, $u_i \in V$ for $i=1, \ldots , k+1$, is an {\em undirected path of length $k$}, if $\{u_i,u_{i+1}\} \in E$ for $i=1,\ldots,k$. The vertices $u_1$ and $u_{k+1}$ are the {\em end vertices} of undirected path $p$. The sequence $(u_1,\ldots,u_{k+1})$ is a {\em directed path of length $k$}, if $(u_i,u_{i+1}) \in E$ for $i=1,\ldots,k$. Vertex $u_1$ is the start vertex and vertex $u_{k+1}$ is the end vertex of the directed path $p$. A path $p$ is a {\em simple path} if all vertices are mutually distinct.

An undirected graph $G$ is {\em connected} if there is a path between every pair of vertices. The {\em distance} $d_G(u,v)$ between two vertices $u,v$ in a connected undirected graph $G$ is the smallest integer $k$ such that there is a path of length $k$ between $u$ and $v$. A {\em connected component} of an undirected graph $G$ is a connected induced subgraph $G'=(V',E')$ of $G$ such that there is no connected induced subgraph $G''=(V'',E'')$ of $G$ with $V'\subseteq V''$ and $|V'|<|V''|$. A vertex $u \in V$ is a {\em separation vertex} of an undirected graph  $G$ if $G|_{V \setminus \{u\}}$ (the subgraph of $G$ induced by $V \setminus \{ u \} $) has more connected components than $G$. Two paths $p_1=(u_1,\ldots,u_k)$ and $p_2=(v_1,\ldots,v_l)$ are {\em vertex-disjoint} if $\{u_2,\ldots,u_{k-1}\} \cap \{v_2\ldots,v_{l-1}\} = \emptyset$. A graph $G=(V,E)$ with at least three vertices is {\em biconnected}, if for every vertex pair $u,v \in V$, $u \not= v$, there are at least two vertex-disjoint paths between $u$ and $v$. A {\em biconnected component} $G'=(V',E')$ of $G$ is an induced biconnected subgraph of $G$ such that there is no biconnected induced subgraph $G''=(V'',E'')$ of $G$ with $V' \subseteq V''$ and $|V'|<|V''|$.

\begin{definition}[Resolving set]
Let $G=(V,E)$ be a connected undirected graph. A vertex set $R \subseteq V$ is a {\em resolving set} for $G$ if for every vertex pair $u,v \in V,$ $u \not = v$, there is a vertex $w \in R$ such that $d_G(u,w) \not= d_G(v,w)$. 
The set $R$ is a {\em minimum resolving set for $G$}, if there is no resolving set $R'\subseteq V$ for $G$ with $|R'| < |R|$.  A connected undirected graph $G=(V,E)$ has {\em metric dimension} $k\in\mathbb{N}$ if $k$ is the smallest positive integer such that there is a resolving set for $G$ of size $k$.
\end{definition}

\begin{definition}
\label{D01}
Let $G=(V,E)$ be a connected undirected graph.
\begin{enumerate}
\item {\bf (leg, root, leaf, hooked leg, ordinary leg)}
A path $p = (u_1,\ldots,u_k)$, $k \ge 2$, of $G$ is a {\em leg}, if vertex $u_1$ has degree one, the vertices $u_{2},\ldots,u_{k-1}$ have degree $2$, and vertex $u_k$ has degree $\geq 3$ in $G$.
Vertex $u_k$ is called the {\em root} of $p$. Vertex $u_1$ is called the {\em leaf} of $p$.
A leg is called a {\em hooked leg}, if the removal of its root separates $G$ into exactly two connected components, i.e. the edges at root $u_k$ without edge $\{u_{k-1},u_k\}$ belong to exactly one biconnected component.
A leg is called an {\em ordinary leg}, if it is not a hooked leg, i.e. if the removal of its root separates $G$ into more than two connected components.
\item {\bf (bridge)}
An edge $e \in E$ is called a {\em bridge} if $H=(V,E \setminus \{e\})$ is not connected and if $e$ is not an edge between two vertices of one and the same leg.
\item {\bf (extended biconnected component (EBC))}
A {\em biconnected component} $H=(V_H,E_H)$ of $G$ extended by the subgraphs of $G$ induced by the vertices of the hooked legs with roots in $V_H$ is an {\em extended biconnected component (EBC)} of $G$.
\item {\bf (component)} Every subgraph induced by the vertices of an ordinary leg, every subgraph induced by the two vertices of a bridge, and every EBC is called a {\em component} of $G$.
\item {\bf (amalgamation vertex)} Separation vertices of $G$ that belong to at least two components, i.e. separation vertices without the degree two vertices of the legs and roots of the hooked legs are called {\em amalgamation vertices}.

\end{enumerate}  
\end{definition}

Every undirected graph $G=(V,E)$ can be decomposed into legs, bridges and EBCs. This decomposition is a unique and edge-disjoint partition of $G$.

\begin{figure}[]  
\center 
\includegraphics[width=0.9\textwidth]{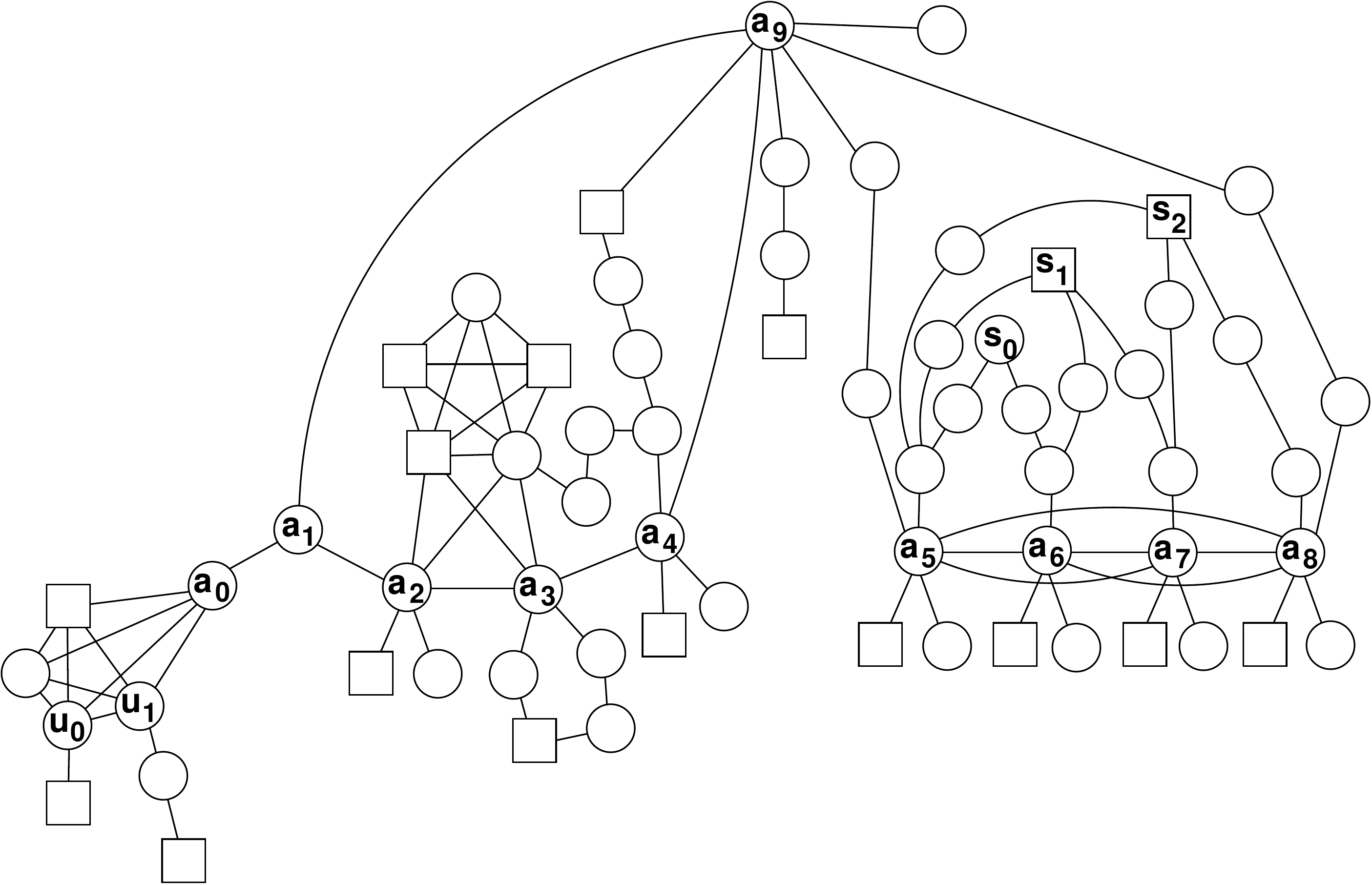}
\caption{Graph $G=(V,E)$ with ten amalgamation vertices ($a_0,\ldots ,a_9$), two hooked legs (at roots $u_0$ and $u_1$), $14$ ordinary legs (two legs at each of the roots $a_2$, $a_4$, $a_5$, $a_6$, $a_7$, $a_8$ and $a_9$), one bridge ($\{a_0,a_1\}$) four EBCs and $19$ components.  The set of vertices that are drawn as squares is a minimum resolving set for $G$. See Figure \ref{fullExampleTree} for a DEBC-tree of $G$ with root $a_9$.}
\label{fullExampleGraph}
\end{figure}

\begin{definition}[EBC-tree]
Let $G=(V,E)$ be a connected undirected graph.
\begin{enumerate}

\item The {\em EBC-tree} $T=(V_T,E_T)$ for $G$ is a tree with two types of nodes called {\em c-nodes} (nodes for the components of $G$) and {\em a-nodes} (nodes for the amalgamation vertices of $G$). $T$ has a c-node for every component of $G$. The vertex set of the corresponding component of $G$ represented by a c-node $c$ is denoted by ${\cal V}(c)$. $T$ has an a-node for every amalgamation vertex of $G$. The amalgamation vertex represented by a-node $a$ is denoted by \textnu $(a)$. 
Let $V_c$ be the set of c-nodes and $V_a$ be the set of a-nodes of $T$. Then $V_T=V_c \cup V_a$ and $E_T$ is the set of all edges $\{c,a\}$ with $c \in V_c$, $a \in V_a$ and \textnu $(a) \in {\cal V}(c) $.
\end{enumerate}
\end{definition}

Note that in the EBC-tree all leaves are c-nodes and there is no edge between two a-nodes and no edge between two c-nodes. All ordinary legs are represented by leaves, all bridges are represented by inner c-nodes, and all EBCs are represented by leaves or inner c-nodes.

\begin{definition}[DEBC-tree]
Let $G=(V,E)$ be a connected undirected graph.

\begin{enumerate}
 
\item For the EBC-tree $T=(V_T,E_T)$ for $G$ and a node $r\in V_T$ let $\overrightarrow{T} := (V_T,\overrightarrow{E}_T)$ be the {\em directed EBC-tree (DEBC-tree) with root $r$} that is defined as follows: $\overrightarrow{E_T}$ contains exactly one directed edge for every undirected edge of $E_T$ such that for every node $u \in V_T$ there is a directed path to root $r$, i.e. all edges are directed from the leafs towards the root.

\item For a node $u \in V_T$, let $\overrightarrow{T}(u)$ be the subtree of $\overrightarrow{T}$ induced by all nodes $v$ for which there is a directed path from $v$ to $u$ in $\overrightarrow{T}$. The root of $\overrightarrow{T}(u)$ is $u$.

\item For a subtree $\overrightarrow{T}(u)$ let $V_T'$ be the set of c-nodes of $\overrightarrow{T}(u)$ and
${\cal V}(V_T') :=  \bigcup_{v \in V_T'}{\cal V}(v)$. Then $G[u] := G|_{{\cal V}(V_T')}$ is the subgraph of $G$ induced by the vertices of ${\cal V}(V_T')$.
\end{enumerate}
\end{definition}

$G[u]$ is the subgraph of $G$ represented by $\overrightarrow{T}(u)$. It is not necessary to refer to the a-nodes of $\overrightarrow{T}(u)$, because the vertices of $G$ that are represented by the a-nodes are also represented by the c-nodes since for every a-node $a$ there is a c-node $c$ such that \textnu $(a) \in {\cal V}(c) $.\\
Note that the EBC-tree and the DEBC-tree of $G$ can be constructed in linear time with the help of any linear time algorithm for finding the biconnected components and bridges of $G$.

\begin{figure}[]
\center
\includegraphics[height= 0.45 \textheight]{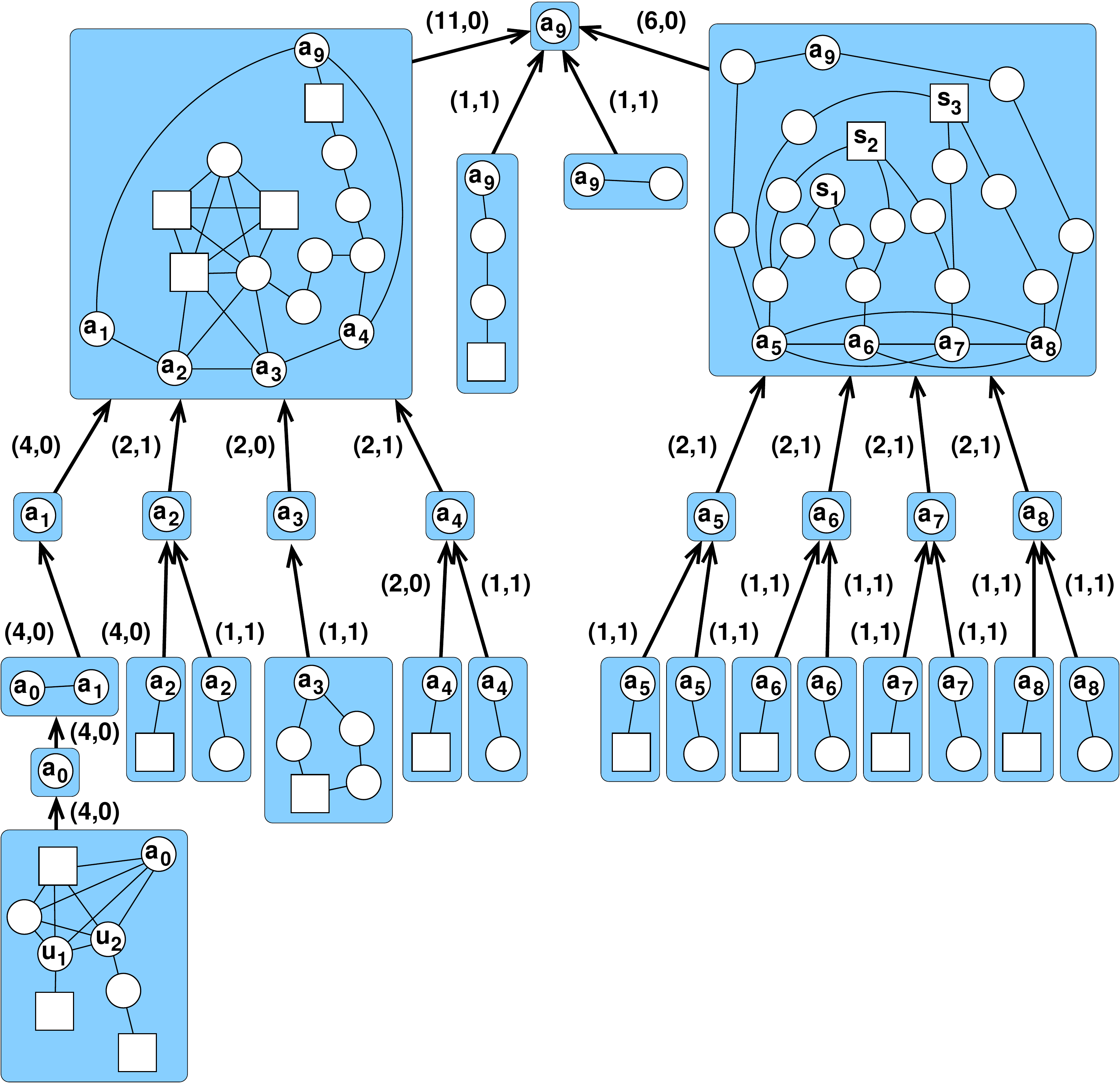} 
\caption{A DEBC-tree $\protect\overrightarrow{T} = (V_T,\protect\overrightarrow{E}_T)$ at root $a_9$ for a graph $G$ with 10 amalgamation vertices $a_0,\ldots ,a_9$, two hooked legs (at roots $u_0$ and $u_1$), $14$ ordinary legs (two legs at each of the roots $a_2$, $a_4$, $a_5$, $a_6$, $a_7$, $a_8$ and $a_9$), one bridge ($\{a_0,a_1\}$) four EBCs and $19$ components.  The vertices that are drawn as squares build a minimum resolving set for $G$. The vertices of $\protect\overrightarrow{T}$ ($19$ c-nodes and 10 a-nodes) are drawn as blue boxes, the directed edges as black arrows. For a c-node $c$ the blue box for $c$ contains the subgraph of $G$ induced by ${\cal V}(c)$ and for an a-node $a$ the blue box for $a$ contains the vertex \textnu $(a)$ of $G$.}
\label{fullExampleTree}
\end{figure}

\section{Computing the metric dimension based on a graph decomposition}

Without loss of generality we will use from now on the following assumptions: 

\begin{enumerate}
 \item $G=(V,E)$ is a connected undirected, but not biconnected graph.
 \item $\overrightarrow{T}= (V_T, \overrightarrow{E}_T)$ is the DEBC-tree for $G$ with root $r$.  \item $V_a$ is the set of a-nodes of $\overrightarrow{T}$ and $V_c$ is the set of c-nodes of $\overrightarrow{T}$.
 \item Root $r \in V_a$ is an a-node.
 \item Root $r$ has at least two children (because $G$ is not biconnected).
\end{enumerate}

First we will describe the general idea of how to compute the metric dimension of $G$. The idea is based on dynamic programming.

\begin{property}\label{properties}
For every subtree $\overrightarrow{T}(v)$, $v \in V_T$, of $\overrightarrow{T}$ we compute an information $h(v)$ satisfying the following properties:
\begin{enumerate}
  \item For every a-node $a \in V_a$ with children $c_1, \ldots, c_k \in V_c$, $k \geq 1$, the information $h(a)$ can efficiently be computed from $h(c_1), \ldots ,h(c_k)$.
  \item For every c-node $c \in V_c$ with children $a_1, \ldots, a_k \in V_a$, $k \geq 0$, the information $h(c)$ can efficiently be computed from $h(a_1), \ldots, h(a_k)$ and $G_{|{\cal V}(c)}$.
  \item The metric dimension of $G[r]$ can efficiently be computed from $h(r)$. 
\end{enumerate}
\end{property}

These properties allow an efficient bottom-up processing of $\overrightarrow{T}$ as follows: 
We start by computing $h(c)$ for every leaf $c$ of $\overrightarrow{T}$. Since the leafs are c-nodes without children we only need the subgraph $G_{|{\cal V}(c)}$ of $G$. For every inner a-node $a$ with children $c_1, \ldots ,c_k \in V_c$ we compute $h(a)$ from $h(c_1), \ldots ,h(c_k)$. For this we don't need any information from $G$. For every inner c-node $c$ with children $a_1, \ldots, a_k \in V_a$ we compute $h(c)$ from $h(a_1), \ldots, h(a_k)$ and additionally $G_{|{\cal V}(c)}$. Finally we compute the metric dimension of $G$ from $h(r)$.
 
Before we define $h(v)$ we need a few more definitions.

\begin{definition}[Gate Vertex] 
\label{GateVertex}
Let $A \subseteq V$ be a set of vertices. A vertex $v \in V$ is an {\em $A$-gate} of $G$, if there is a vertex $u \in V \setminus \{v\}$, such that for all $w \in A$ the equation $d_G(u,w)=d_G(u,v)+d_G(v,w)$ holds.
Vertex $u$ is called an {\em out-vertex} for $A$-gate $v$.
\end{definition}

Intuitively this definition means that $v$ is an $A$-gate if there is a vertex $u \in V \setminus \{v\}$ that has a shortest path to any vertex $a \in A$ that passes $v$.

\begin{observation}\label{adjOut}
 Let $A \subseteq V$ and $v \in V$ be an $A$-gate of $G$. Then there is an out-vertex $u \in V$ adjacent to $v$.
\end{observation}

\begin{observation}\label{outOnPath}
 Let $A \subseteq V$, $v, u_1, u_2 \in V$ and $v$ be an $A$-gate of $G$. If $u_1$ and $u_2$ are two out-vertices for $A$-gate $v$ with the same distance to $v$, i.e. $d_G(u_1,v) = d_G(u_2,v)$, then both vertices $u_1$ and $u_2$ have the same distance to all vertices of $A$. In this case $A$ is not a resolving set for $G$. Conversely, if $A$ is a resolving set for $G$ then all out-vertices have a different distance to $A$-gate $v$. A closer look shows that if $A$ is a resolving set for $G$ all out-vertices for $A$-gate $v$ are on a shortest path between $v$ and the out-vertex with longest distance to $v$, see also Figure \ref{outVertex}.  
\end{observation}

\begin{figure}[]
\center
\includegraphics[width=0.8 \textwidth]{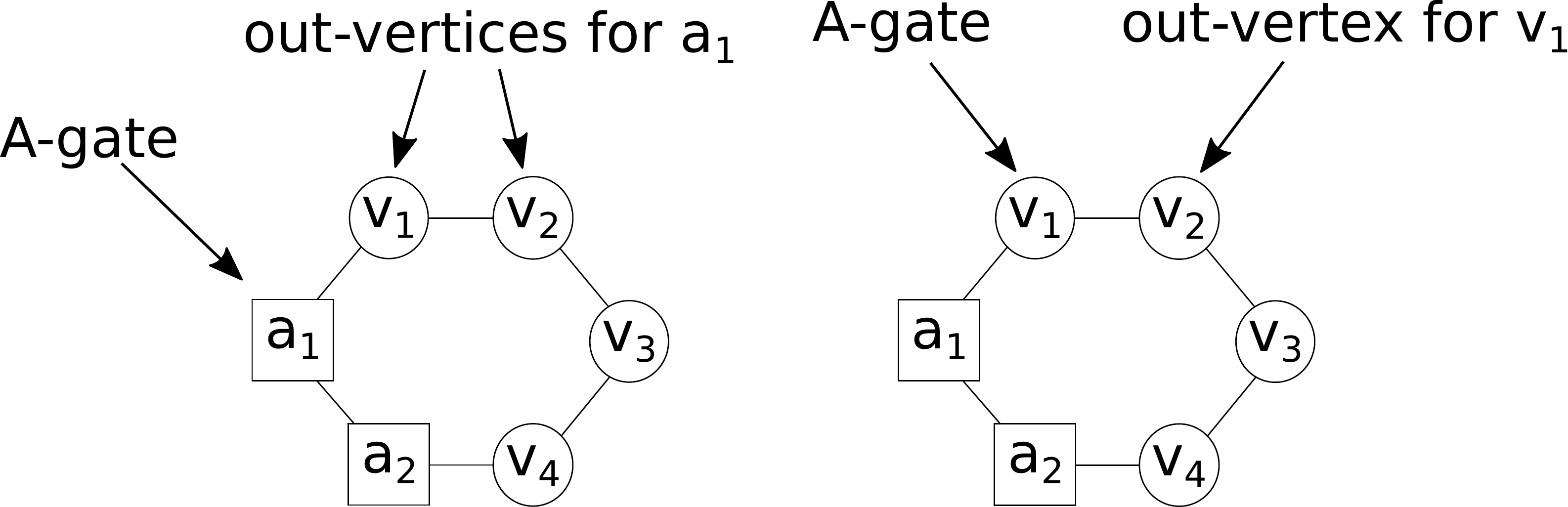}
\caption{\label{outVertex} The Figure shows graph $C_6=(V,E)$ drawn two times. The set $A=\{a_1,a_2\} \subseteq V$ is a resolving set for $C_6$. The vertex $a_1$ is an $A$-gate with out-vertices $v_1$ and $v_2$ as shown on the left side. The vertex $v_1$ is an $A$-gate with out-vertex $v_2$ as shown on the right side. Notice that there is no other $A$-gate in $C_6$.}
\end{figure}

\begin{definition}[$v$-resolving set, non-gate-$v$-resolving set]
Let $v \in V.$ 
\begin{enumerate}
 \item A {\em $v$-resolving set} for $G$ is a resolving set $R$ for $G$ with $v \in R$.
 \item A {\em minimum $v$-resolving set} for $G$ is a resolving set $R$ for $G$ with $v \in R$ such that there is no $v$-resolving set $R'$ for $G$ with $|R'| < |R|$.
 \item A {\em non-gate-$v$-resolving set} for $G$ is a $v$-resolving set $R$ for $G$ with $v \in R$ and $v$ is not an $R$-gate in $G$.
 \item A {\em minimum non-gate-$v$-resolving set} for $G$ is a $v$-resolving set $R$ for $G$ with $v \in R$ and $v$ is not an $R$-gate in $G$, such that there is no non-gate-$v$-resolving set $R'$ for $G$ with $|R'| < |R|$.
\end{enumerate}

\end{definition}

Note that a minimum $v$-resolving set is not necessarily a minimum resolving set. It is possible that no minimum resolving set contains $v$. Also a minimum non-gate-$v$-resolving set is not necessarily a minimum $v$-resolving set. It is possible that in every minimum $v$-resolving set $R$ vertex $v$ is an $R$-gate.

\begin{lemma} \label{eq}
Let $v \in V$. Let $R_1 \subseteq V$ a minimum resolving set for $G$, $R_2 \subseteq V$ a minimum $v$-resolving set for $G$, and $R_3 \subseteq V$ a minimum non-gate-$v$-resolving set for $G$, then  $|R_2| \leq |R_1|+1$ and $|R_3| \leq |R_2|+1$.
\end{lemma}

\begin{proof}
 $|R_2| \leq |R_1|+1$ : If $R_1$ contains vertex $v$, $R_1$ is already a minimum $v$-resolving set. If $v \not \in R_1$ the set $R_2 := R_1 \cup \{v\}$ is a $v$-resolving set.\\
 $|R_3| \leq |R_2|+1$ : If $v$ is not an $R_2$-gate, $R_2$ is already a minimum non-gate-$v$-resolving set. If $v$ is an $R_2$-gate there are out-vertices $u_1, \ldots ,u_k$ for $v$. From Observation \ref{outOnPath} we know that these out-vertices are on a shortest path between $v$ and the out-vertex with longest distance to $v$, so for any $i \in \{1, \ldots , k\}$, $k\geq 1$, the set $R_2 \cup \{u_i\}$ is a non-gate-$v$-resolving set. 
\end{proof}

Figure \ref{sharpBounds} shows that these bounds are tight.

\begin{figure}[]
\center 
\includegraphics[width=0.5\textwidth]{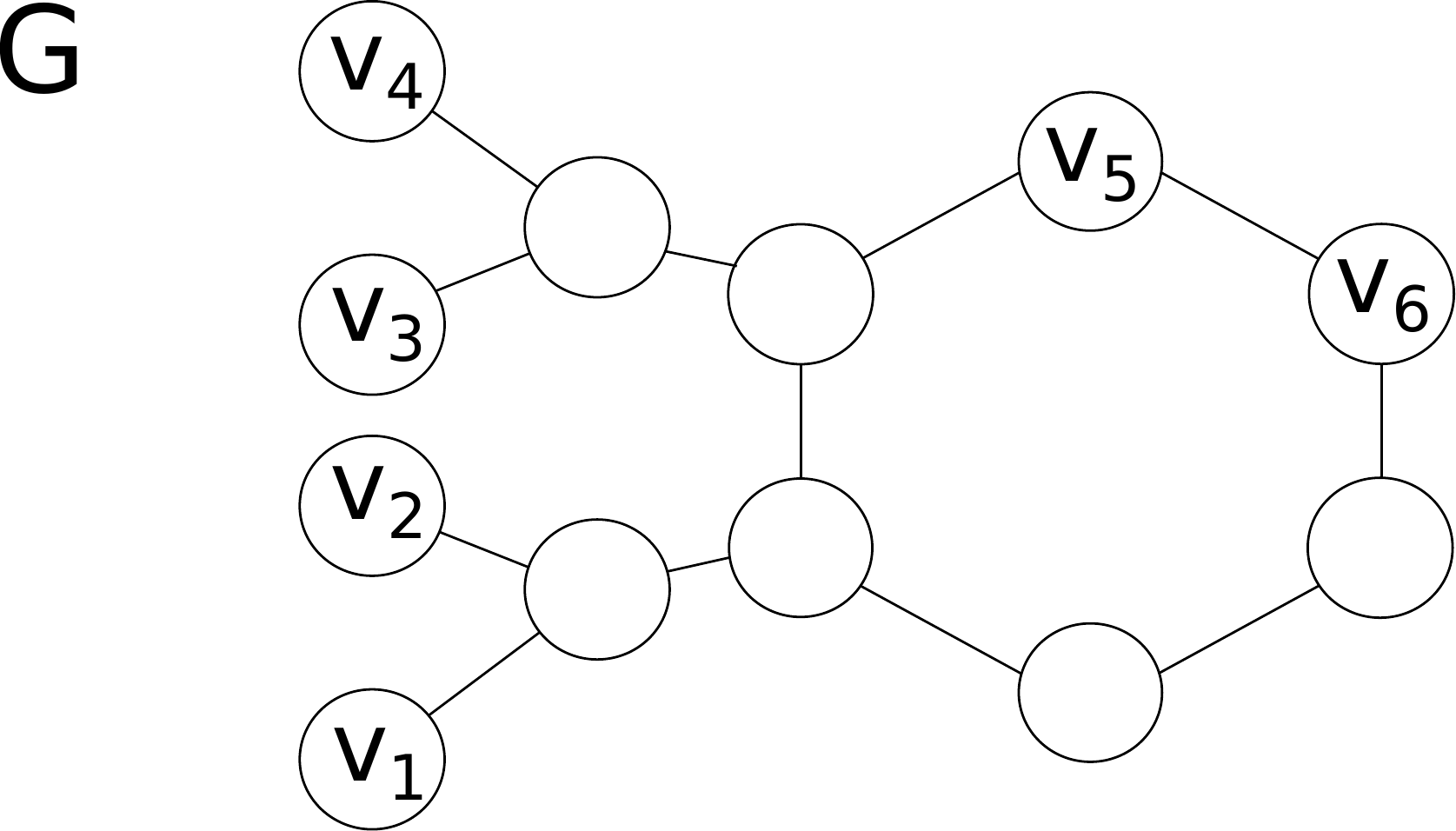}
\caption{
\label{sharpBounds} There are four minimum resolving sets for $G$, namely $\{v_1, v_3\}$, $\{v_1, v_4\}$,$\{v_2, v_3\}$, and $\{v_2, v_4\}$. None of them contains $v_5$. Therefore every minimum $v_5$-resolving set contains three vertices, i.e. $\{v_1, v_3, v_5\}$, $\{v_1, v_4, v_5\}$, $\{v_2, v_3, v_5\}$, and $\{v_2, v_4, v_5\}$. In every minimum $v_5$-resolving set $R$ the vertex $v_5$ is an $R$-gate. The corresponding out-vertex is $v_6$. Therefore every minimum non-gate-$v_5$-resolving set contains four vertices, i.e. $\{v_1,v_3,v_5,v_6\}$, $\{v_1,v_4,v_5,v_6\}$,$\{v_2,v_3,v_5,v_6\}$, and $\{v_2,v_4,v_5,v_6\}$.
}
\end{figure}

\begin{lemma}\label{ancorfreeSet}
 Let $s \in V$ be a separation vertex and $V_1, \ldots , V_k$, $k>2$, be the vertex sets of the connected components of $G_{|V \setminus \{s\}}$. Let $R$ be a resolving set for $G$. Then there is at most one $i \in \{1, \ldots ,k\}$ such that $V_i \cap R = \emptyset$.
\end{lemma}

\begin{proof}
 Assume there are two sets $V_1, V_2$ such that $V_1 \cap R = \emptyset$ and $V_2 \cap R = \emptyset$. Consider a vertex $v_1 \in V_1$ adjacent to $s$ and a vertex $v_2 \in V_2$ adjacent to $s$. For any $r \in R$ we have $d(v_1,r) = d_G(v_1,s) + d_G(s,r) = d_G(v_2,s) + d_G(s,r) = d_G(v_2,a)$. This contradicts the assumption that $R$ is a resolving set.
\end{proof}

\begin{lemma}\label{withoutSep}
 Let $s \in V$ be a separation vertex and $V_1, \ldots , V_k$, $k>1$, be the vertex sets of the connected components of $G_{|V \setminus \{s\}}$ such that if $k=2$ in every resolving set $R$ for $G$ there is a vertex $r \in V_1 \cap R$ and a vertex $r \in V_2 \cap R$. Let $A \subseteq V$ with $s \in A$. If $A$ is a resolving set for $G$ then $A' = A \setminus \{s\}$ is resolving set for $G$.
\end{lemma}

\begin{proof} 
 Consider two vertices $u,v \in V$ that are separated by $s$. We will show, that there is a vertex $a \in A$, $a \neq s$, that separates $u$ and $v$.\\
 Assume $u$ and $v$ are both in the same set $V_i$, $i \in \{1, \ldots , k\}$. Then, by assumption and Lemma \ref{ancorfreeSet}, there is a vertex $a \in A \cap V_j$, $j \neq i$. Since every path from $u$ to $a$ or $v$ to $a$ contains vertex $s$ and $s$ separates $u$ and $v$, vertex $a$ does the same.\\
 Assume $u$ and $v$ are in different sets. Without loss of generality let $v_i \in V_i$, $v_j \in V_j$, $i \neq j$, and $d_G(u,s) < d_G(v,s)$.\\ 
 {\bf Case 1:} $V_i \cap A = \emptyset$\\
 If $V_i \cap A = \emptyset$ then by assumption $G_{|V\setminus \{s\}}$ has at least three components. By Lemma \ref{ancorfreeSet} there is a component $V_l$, $l \neq i$, $l \neq j$, such that $V_l \cap A \neq \emptyset$. Let $a \in V_l \cap A$, then we have
 $$d_G(u,a) = d_G(u,s) + d_G(s,a) < d_G(v,s) + d_G(s,a) = d_G(v,a).$$
 {\bf Case 2:} $V_i \cap R \neq \emptyset$\\
 Then we have for any $a \in V_i$
 $$d_G(u,a) \leq d_G(u,s) + d_G(s,a) < d_G(v,s) + d_G(s,a) = d_G(v,a).$$
\end{proof}

\begin{lemma}\label{maxOneGate}
Let $s\in V$ be a separation vertex and $V_1, \ldots ,V_k$, $k>1$, be the vertex sets of the connected components of $G|_{V\setminus\{s\}}$ such that if $k=2$ in every resolving set $R$ for $G$ there is a vertex $r \in V_1 \cap R$ and a vertex $r \in V_2 \cap R$. Let $A \subseteq V$, $G_i := G|_{V_i\cup\{s\}}$, and $A_i := (A \cap V_i) \cup \{s\}$. $A$ is a minimum resolving set for $G$ if and only if for all $i \in \{1, \ldots , k\}$ $A_i$ is a minimum $s$-resolving set for $G_i$ and there is at most one $i \in \{1, \ldots ,k\}$ such that $s$ is an $A_i$-gate in $G_i$.
\end{lemma}

\begin{proof}
$\Rightarrow$: Let $A$ be a resolving set for $G$. We will show that (1) for every $i \in \{1, \ldots , k\}$ the set $A_i$ is an $s$-resolving set for $G_i$ and (2) there is at most one index $i \in \{1, \ldots ,k\}$ such that $s$ is an $A_i$-gate in $G_i$.\\
(1): Consider two vertices $u, v \in V_i$. Obviously $A_i$ is a resolving set for those pairs of vertices that are separated by a vertex $a \in A\cap V_i$.
Assume that for a pair of vertices $u,v$ there is no vertex in $A \cap V_i$ that separates them, i.e. $u,v$ are separated by a vertex $a \in A \cap V_l$, $l \neq i$. Since all paths from $u$ to $a$ and $v$ to $a$ contain vertex $s$, we have $$d_G(u,a) = d_G(u,s) + d_G(s,a) \text{ and } d_G(v,a) = d_G(v,s) + d_G(s,a).$$ Therefore $d_G(u,a) \neq d_G(v,a)$ implies $d_G(u,s) \neq d_G(v,s)$, i.e. $s \in A_i$ separates $u$ and $v$.\\
(2): Assume there are two indices $i,j \in\{1, \ldots, k\}$, $i\neq j$, for which $s$ is an $A_i$-gate in $G_i$ and an $A_j$-gate in $G_j$. Then there are out-vertices $v_{i_1}, \ldots , v_{i_l}\in V_i$ for $A_i$-gate $s$ in $G_i$ and out-vertices $v_{j_1}, \ldots , v_{j_m} \in V_j$ for $A_j$-gate $s$ in $G_j$, $l,m>0$. From Observation \ref{outOnPath} we know that the out-vertices $v_{i_1}, \ldots , v_{i_l}\in V_i$ are on a shortest path between $s$ and the out-vertex of $v_{i_1}, \ldots , v_{i_l} \in V_i$ with longest distance to $s$. The same holds for the out-vertices $v_{j_1}, \ldots , v_{j_m} \in V_j$. Without loss of generality let $v_{i_1} \in V_i$ and $v_{j_1} \in V_j$ be the out-vertices adjacent to $s$, i.e. $d_G(v_{i_1},s) = d_G(v_{j_1},s) = 1$. Due to the decomposition of $G$ with separation vertex $s$ there is for any $a \in A$ a shortest path between $v_{i_1}$ and $a$ and between $v_{j_1}$ and $a$ via vertex $s$. This holds for all vertices $a \in A$ and not only for the vertices $a \in V_i$ or $a \in V_j$, respectively. It follows that

\[ \forall a \in A: \quad d_G(v_{i_1},w)= d_G(v_{i_1},s) + d_G(s,w) = d_G(v_{j_1},s) + d_G(s,w) = d_G(v_{j_1},w)  \]

This contradicts the assumption that $A$ is a resolving set for $G$, see also Figure \ref{separationGate}.\\ \null

$\Leftarrow$: Let $A_i$ be an $s$-resolving set for $G_i$ and let there be at most one $i \in \{1, \ldots ,k\}$ such that $s$ is an $A_i$-gate in $G_i$. We will proof that $A$ is a resolving set for $G$.\\
Let $A' = \bigcup_{i=1}^k A_i$. Obviously $A'$ separates all pairs of vertices that are in the same set $V_i$. Consider two vertices $u,v$ that are in different set. Without loss of generality let $v_i \in V_i$ and $v_j \in V_j$, $i,j \in \{1, \ldots ,k\}$, $i \neq j$ and $G_i$ the component in which $s$ is not an $A_i$-gate for $G_i$. This implies that there is a vertex $a \in A_i$, $a \neq s$ such that there is no shortest path from $v_i$ to $a$ via $s$. Since every path from $v_j$ to $a$ contains $s$, vertex $a$ separates $u$ and $v$.
This implies that $A'$ is a resolving set for $G$ and with Lemma \ref{withoutSep} we get that that $A = A' \setminus \{s\}$ is a resolving set for $G$.\\ \null

Now we will show the minimality. Assume that $A$ is a minimal resolving set for $G$ and $A_i$ is not a minimal $s$-resolving set, i.e. there is set $A_i'$ with $|A_i'| < |A_i|$ that contains $s$ and resolves all vertices from $G_i$. Consider the set $A' := \bigcup_{j=1}^k A_j \setminus \{s\}$. As shown above $A'$ is a resolving set for $G$ and we have with $|A'| < |A|$. This contradicts the assumption that $A$ is minimal.\\
Assume that $A_i$ is a minimal $s$-resolving set and there is at most one $j \in \{1, \ldots ,k\}$ such that $s$ is an $A_j$-gate in $G_j$, but $A$ is not minimal, i.e. there is a resolving set $A'$ for $G$ with $|A'| < |A|$. This implies that there is at least one $i \in \{1, \ldots, k\}$ such that $A_i':=(A'\cap V_i) \cup \{s\}$ is a resolving set for $G_i$ and $|A_i'| < |A_i|$. This contradicts the assumtion that $A_i$ is minimal.
\end{proof}

\begin{figure}[] 
\center
\includegraphics[width=\textwidth]{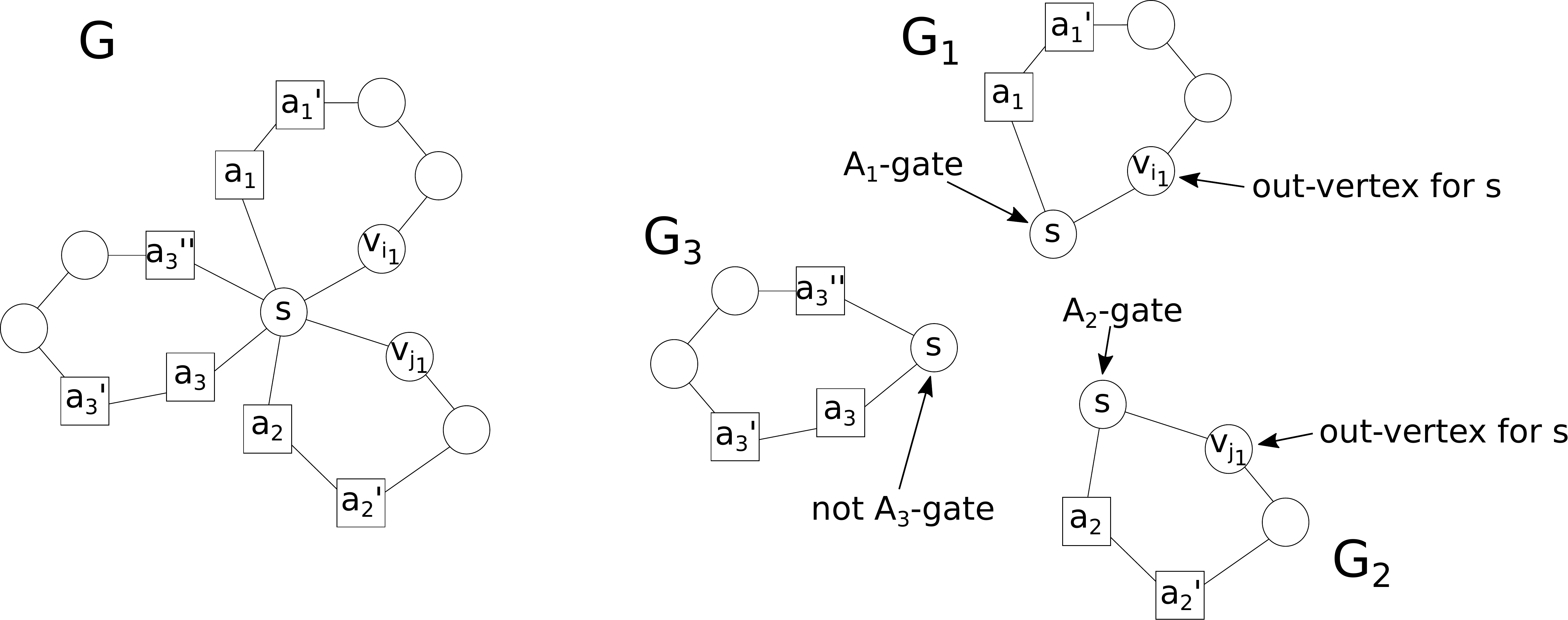}
\caption{\label{separationGate} Graph $G=(V,E)$ with three connected components and separation vertex $s \in V$ on the left side and graphs $G_1$, $G_2$ and $G_3$ on the right side. The vertices of the resolving sets $A_1=\{a_1, a_1'\}$ for $G_1$, $A_2=\{a_2, a_2'\}$ for $G_2$ and $A_3=\{a_3, a_3', a_3''\}$ for $G_3$ are drawn as squares. $s$ is an $A_1$-gate in $G_1$ and an $A_2$-gate in $G_2$, therefore there are two vertices $v_{i_1}$ and $v_{j_1}$ that are not separated by $A:= \bigcup_{i=1}^3 A_i$ and $A$ is not a resolving set for $G$.}
\end{figure}

\begin{lemma}\label{minSepSet}
 Let $s\in V$ be a separation vertex and $V_1, \ldots ,V_k$, $k>1$, the vertex sets of the connected components of $G|_{V\setminus\{s\}}$ such that if $k=2$ in every resolving set $R$ for $G$ there is a vertex $r \in V_1 \cap R$ and a vertex $r \in V_2 \cap R$. Let $A \subseteq V$. If $A$ is a minimum resolving set for $G$, then $A' := A \cup \{s\}$ is a minimum $s$-resolving set for $G$.
\end{lemma}

\begin{proof}
Let $A$ be a minimum resolving set and $A' = A \cup \{s\}$ an $s$-resolving set. From Lemma \ref{withoutSep} we know that $s \notin A$, therefore $|A|<|A'|$. Assume there is a miniumum $s$-resolving set $A''$ with $|A''| < |A'|$. Then, by Lemma \ref{withoutSep}, $A'' \setminus \{s\}$ is a resolving set for $G$ and we have $|A'' \setminus \{s\}| < |A|$. This contradicts the assumption that $A$ is a minimum resolving set for $G$.
\end{proof}

Now we define $h(v)$, $v \in V_T$, as introduced at the beginning of the section.

\begin{definition}~ \label{h}
\begin{enumerate}
 \item Let $a \in V_a$ be an a-node. We define $h(a) := (\alpha, \beta)$, where $\alpha$ is the size of a minimum non-gate-\textnu $(a)$-resolving set for $G[a]$ and $\beta$ is the size of a minimum \textnu $(a)$-resolving set for $G[a]$.
 \item Let $c \in V_c$ be a c-node with father $a \in V_a$ in $\overrightarrow{T}$. We define $h(c) := (\alpha, \beta)$, where $\alpha$ is the size of a minimum non-gate-\textnu $(a)$-resolving set for $G[c]$ and $\beta$ is the size of a minimum \textnu $(a)$-resolving set for $G[c]$.
\end{enumerate}
\end{definition}

To get familiar with this definition we will investigate the smallest possible values for $\alpha$ and $\beta$. For an arbitrary node $v \in \overrightarrow{T}$ with father $w \in \overrightarrow{T}$ we have $h(v) \leq h(w)$, i.e. the $i$-th component of $h(v)$ is less or equal than the $i$-th component of $h(w)$, since $G[v]$ is a subgraph of $G[w]$. Therefore we will first have a look at the leafs of $\overrightarrow{T}$, which are by definition c-nodes and afterwards at the fathers of the leafs, which are by definition a-nodes.
For a leaf $c$ with father $a$ the graph $G[c]$ is either an EBC, or an ordinary leg. Let $G[c]$ be an ordinary leg. Then vertex \textnu $(a) \in G[c]$ resolves all vertices in $G[c]$, so $\beta =1$. Since \textnu $(a)$ is a $\{\text{\textnu }(a)\}$-gate in $G[c]$ every minimum non-gate-\textnu $(a)$-resolving set contains another arbitrary vertex. Therefore $\alpha = 2$. Let $G[c]$ be an EBC, then every resolving set for $G[c]$ contains at least two vertices. Therefore $h(c) \geq (2,2)$.\\ \null  

For an a-node $a$ that has only leafs as children the graph $G[a]$ consists of EBCs and paths, that are connected by the separation vertex \textnu $(a)$. Note that if $a$ has exactly one child $c$ the graph $G[c]$ is not an ordinary leg, since this contradicts the decomposition of $G$ into EBCs, ordinary legs, and bridges. Thus every minimum \textnu $(a)$-resolving set for $G[a]$ contains at least two vertices and the smallest values $\alpha$ and $\beta$ for an a-node $a$ are $h(a)=(2,2)$, what leads to the following observation:

\begin{observation}
Let $S$ be a minimum resolving set for $G$. For any a-node $a \in V(\overrightarrow{T})$ the subgraph $G[a]$ contains at least one resolving node, i.e. $S \cap V(G[a]) \neq \emptyset$.
 \label{G[a]Ancor}
\end{observation}

We will now show that this definition satisfys the properties in \ref{properties}.

\begin{theorem}\label{theorem1}
 For every a-node $a \in V_a$ with children $c_1, \ldots, c_k \in V_c$, $k \geq 1$, $h(a)$ can be computed from $h(c_1), \ldots ,h(c_k)$.
\end{theorem}

\begin{proof}
 $k=1$: If $a$ has exactly one child $c$ then $h(a) = h(c)$. Since \textnu $(a) \in {\cal V} (c)$ and $c$ is the only child of $a$, we can follow that $G[a] = G[c]$. Therefore a minimum non-gate-\textnu $(a)$-resolving set for $G[c]$ is also a minimum non-gate-\textnu $(a)$-resolving set for $G[a]$. The same holds for a minimum \textnu $(a)$-resolving set.\\ \null
 
 $k\geq2$: Let $h(c_i) = (\alpha_i, \beta_i)$, $i \in \{1, \ldots, k\}$. Then $h(a) = (\alpha, \beta)$ with 
 \begin{itemize}
  \item $\alpha = (\sum_{i=1}^k \alpha_i) - (k-1) $
  \item $\beta = 
	\left\{\begin{array}{ll}
	  \alpha, & \text{if } \beta_i = \alpha_i \enspace \forall i  \\
	  \alpha - 1, & \text{else}
        \end{array}\right.$
 \end{itemize}

 Let $A_i$ be a minimum non-gate-\textnu $(a)$-resolving set for $G[c_i]$ and thus $|A_i| = \alpha_i$. Then every $A_i$ is also a minimum \textnu $(a)$-resolving set and there is no $i$, $i \in \{1, \ldots ,k\}$, such that $A_i$ is a \textnu $(a)$-gate in $A_i$. With the help of Lemma \ref{maxOneGate} it follows that $A := \bigcup_{i} A_i \setminus \{\text{\textnu} (a)\}$ is a minimum resolving set for $G[a]$ and from Lemma \ref{minSepSet} it follows that $A' := A \cup \{ \text{\textnu}(a) \}$ is a minimum \textnu $(a)$-resolving set for $G[a]$. Since there is no index $i$ such that \textnu $(a)$ is an $A_i$-gate in $G[c_i]$, it follows that \textnu $(a)$ is not an $A'$-gate in $G[a]$. Since \textnu $(a) \in A_i$ forall $i$ we have $|A'| = \alpha = \sum_{i} \alpha_i - (k-1) $.
 
 In a minimum \textnu $(a)$-resolving set $A''$ for $G[a]$ there is at most one index $i$ such that \textnu $(a)$ is a $A'' \cap V(G[c_i])$-gate in $G[c_i]$, otherwise there would be two vertices $u,v \in V(G[a])$ such that for every $w \in A''$ there is a shortest path to $w$ via \textnu $(a)$, i.e. $A''$ is not a resolving set (see Lemma \ref{ancorfreeSet}). Therefore $\beta = \alpha - 1$, if there is one index i such that $\beta_i < \alpha_i $ and $\beta = \alpha$ else.
 \end{proof}

\begin{theorem}\label{theorem2}
 For every c-node $c \in V_c$ with father $a \in V_a $ and children $a_1, \ldots, a_k \in V_a$, $k \geq 0$,  $h(c)$ can be computed from $h(a_1), \ldots, h(a_k)$ and $G_{|{\cal V}(c)}$.
\end{theorem}

To proof this theorem, we need the following lemma:

\begin{lemma}\label{thLemma}
 Let $c \in V_c$ be a c-node with father $a \in V_a$ and children $a_1, \ldots, a_k \in V_a$, $k \geq 0$. Let $R \subseteq V(G[c])$ with \textnu $(a) \in R$. Let $R_i := (R \cap V(G[a_i])) \cup \{ \text{\textnu }(a_i) \}$, $i \in \{1, \ldots, k\}$, and $R^* := (R \cap {\cal V}(c))) \cup \{ \text{\textnu }(a_i) \; | \; 1 \leq i \leq k \}$. $R$ is a \textnu $(a)$-resolving set for $G[c]$ if and only if 
 \begin{enumerate}
  \item $R_i$ is a resolving set for $G[a_i]$ and
  \item $R^*$ is a \textnu $(a)$-resolving set for $G_{|{\cal V}(c)}$ and 
  \item For every $i \in \{1, \ldots , k\}$ vertex \textnu $(a_i)$ is not an $R_i$-gate in $G[a_i]$ or not an $R^*$-gate in $G_{|{\cal V}(c)}$.
 \end{enumerate}
\end{lemma}

\begin{proof} 
"$\Rightarrow$":
Let $R$ be a \textnu $(a)$-resolving set for $G[c]$. 
\begin{enumerate}
 \item We show that $R_i$ is a resolving set for $G[a_i]$. Let $u,v \in V(G[a_i])$. Pair $u,v$ is either resolved by a vertex $w_i \in V(G[a_i]) \cap R$, or by a vertex $w \in {\cal V}(c) \cap R$, or by a vertex $w_j \in V(G[a_j]) \cap R$, $j \in \{ 1, \ldots , k \}$, $j \neq i$. Since every path from a vertex of $V(G[a_i])$ to $w$ or to $w_j$ contains vertex \textnu $(a_i)$, vertex $w$ or $w_j$ resolves $u,v$ if and only if \textnu $(a_i)$ resolves $u,v$ , see Figure \ref{firstDir1}. Thus, set $R \cap V(G[a_i]) \cup \{ \text{\textnu } (a_i) \} = R_i$ resolves all pairs in $V(G[a_i])$.
 
 \item We show that $R^*$ is a \textnu $(a)$-resolving set for $G_{|{\cal V}(c)}$. Since \textnu $(a) \in R^*$ by definition, we just have to show that $R^*$ is a resolving set. Let $u,v \in {\cal V}(c)$. Pair $u,v$ is either resolved by a vertex $w \in {\cal V}(c) \cap R$ or by a vertex $w_j \in V(G[a_j]) \cap R$, $j \in \{1, \ldots , k\}$. Since every path from a vertex in ${\cal V}(c)$ to $w_j$ contains vertex \textnu $(a_j)$, vertex $w_j$ resolves $u,v$ if and only if \textnu $(a_j)$ resolves $u,v$, see Figure \ref{firstDir2}. Thus, set $R \cap {\cal V}(c) \cup \{ \text{\textnu } (a_j) \; | \;  1\leq j \leq k, k \geq 0 \} = R^*$ resolves all pairs in ${\cal V}(c)$.
 
 \item We show that for every $i \in \{1, \ldots , k\}$ vertex \textnu $(a_i)$ is not an $R_i$-gate in $G[a_i]$ or not an $R^*$-gate in $G_{|{\cal V}(c)}$. Assume there is an index $i \in \{1, \ldots ,k\}$ such that vertex \textnu $(a_i)$ is an $R_i$-gate in $G[a_i]$ and an $R^*$-gate in $G_{|{\cal V}(c)}$. From Observation \ref{adjOut} we know that there is an out-vertex $v_i \in V(G[a_i])$ with respect to $R_i$ adjacent to \textnu $(a_i)$ and an out-vertex $v \in {\cal V}(c)$ with respect to $R^*$ adjacent to \textnu $(a_i)$. That is there are two vertices $v_i$ and $v$ with the same distance to \textnu $(a_i)$, such that from both, there is a shortest path to every vertex in $R_i \cup R^*$ that contains \textnu $(a_i)$. This implies that pair $v_i,v$ cannot be resolved by a vertex in $R':= R \cap (V(G[a_i]) \cup {\cal V}(c)$ and therefore must be resolved by another vertex in $R \setminus R'$, i.e. by a vertex in $V(G[a_j])$, $j \in \{1, \ldots, k\}$, $j \neq i$. However, if there is a vertex in $V(G[a_j])$, that resolves pair $v_i,v$, then vertex \textnu $(a_j) \in R^*$ also resolves pair $v_i, v$. This implies that there is no $j \in \{1, \ldots, k\}$, $j \neq i$, such that a vertex in $V(G[a_j])$ resolves $v, v_i$ and therefore $R$ is not a \textnu $(a)$-resolving set for $G[c]$, what contradicts the assumption.\\ \null
 \end{enumerate}

 \begin{figure}[]
 \centering
    \begin{subfigure}[t]{0.45\textwidth}
        \includegraphics[width=\textwidth]{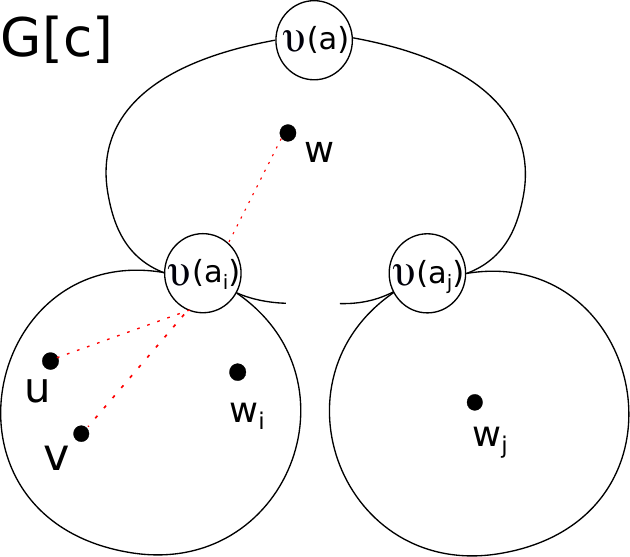}
    \caption{
      \label{firstDir1} The Figure shows graph $G[c]$ with subgraph $G[a_i]$ and $G[a_j]$ and vertex pair $u,v \in V(G[a_i]$ that needs to be resolved. They can be resolved by a vertex $w_i \in  R \cap V(G[a_i])$ or by a vertex $w \in R \cap {\cal V}(c)$ or by a vertex $w_j \in R \cap V(G[a_j])$.}
    \end{subfigure}
    \quad 
     \begin{subfigure}[t]{0.45\textwidth}
        \includegraphics[width=\textwidth]{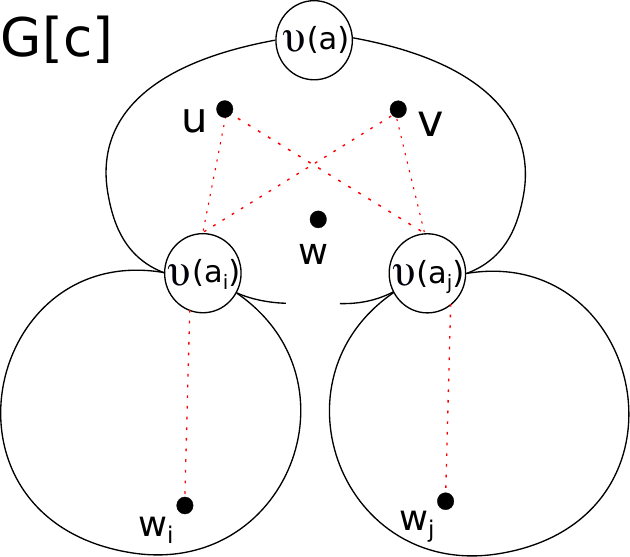}
    \caption{
      \label{firstDir2} The Figure shows graph $G[c]$ with subgraph $G[a_i]$ and $G[a_j]$ and vertex pair $u,v \in {\cal V}(c)$ that needs to be resolved. They can be resolved by a vertex $w_i \in  R \cap V(G[a_i]$ or by a vertex $w \in R \cap {\cal V}(c)$. }
    \end{subfigure} 
     \caption{}
\end{figure}

 "$\Leftarrow$":
Let 
 \begin{enumerate}
  \item $R_i$ be a resolving set for $G[a_i]$ 
  \item $R^*$ be a \textnu $(a)$-resolving set for $G_{|{\cal V}(c)}$  
  \item for every $i \in \{1, \ldots , k\}$ vertex \textnu $(a_i)$ be not an $R_i$-gate in $G[a_i]$ or not an $R^*$-gate in $G_{|{\cal V}(c)}$.
 \end{enumerate}
 We show that $R$ is a \textnu $(a)$-resolving set.
 $R$ is a \textnu $(a)$-resolving set, if \textnu $(a) \in R$ and if for every pair $u,v \in V(G[c])$ there is a vertex in $R$, that resolves $u$ and $v$. Obviously \textnu $(a) \in R$, because  \textnu $(a) \in R^*$, so we just have to show that $R$ is a resolving set.
 We divide between the follwing four cases:
    \begin{itemize}
      \item [(a)] $u,v \in V(G[a_i])$, $i \in \{1, \ldots , k\}$,
      \item [(b)] $u \in V(G[a_i])$, $v \in V(G[a_j])$, $i,j \in \{1, \ldots , k\}$, $i \neq j$,
      \item [(c)] $u,v \in {\cal V}(c)$,
      \item [(d)] $u \in V(G[a_i])$ and $v \in {\cal V}(c)$, $i \in \{1, \ldots , k\}$.
    \end{itemize}
     
We will have a closer look at all these cases.
 \begin{itemize}
      \item [(a)] Let $u,v \in V(G[a_i])$. Pair $u,v$ is either resolved by a vertex $w_i \in R_i \setminus \{\text{\textnu } (a_i)\} \subseteq R$, or by vertex \textnu $(a_i)$. Since $R \cap {\cal V}(c) \neq \emptyset$ and every pair that is resolved by \textnu $(a_i)$ is also resolved by a vertex in $R \cap {\cal V}(c)$ (\textnu $(a) \in R$), $R$ resolves all pairs $u,v \in V(G[a_i])$, see Figure \ref{secondDir1}.
      
      \item [(b)] Let $u \in V(G[a_i])$ and $v \in V(G[a_j])$. Pair $u,v$ is either resolved by $w_i \in R_i \setminus \{\text{\textnu } (a_i)\} \subseteq R$ or by $w_j \in R_j \setminus \{\text{\textnu } (a_j)\} \subseteq R$, because it is not possible that $u$ and $v$ have the same distance to both, $w_i$ and $w_j$, what can be seen as follows:\\
      Let $a:= d_{G[c]} (u,w_i)$, $b:= d_{G[c]} (u,\text{\textnu }(a_i))$, $c:= d_{G[c]} (w_i,\text{\textnu }(a_i))$,     $a':= d_{G[c]} (v,w_j)$, $b':= d_{G[c]} (v,\text{\textnu }(a_j))$, $c':= d_{G[c]} (w_i,\text{\textnu }(a_j))$ and $d:= d_{G[c]} (\text{\textnu }(a_i),\text{\textnu }(a_j))$, (see Figure \ref{secondDir2}). Since $w_i \neq w_j$ we can follow $d>0$. Assume that neither $w_i$ nor $w_j$ separate $u,v$, i.e. $a = c+d+b'$ and $a'= b+d+c'$. From $a \leq b+c$ and $a\leq b' + c'$ we get $c+d+b' \leq b+c$ and $b+d+c' \leq b'+c'$. Finally we have $d+b' \leq b'-d$ what implies $d \leq -d$ and thus $d=0$. This contradicts the assumption $d>0$. 
      It follows that $\bigcup_{i=1}^k R_i \setminus \{ \text{\textnu } (a_i)\} \subseteq R$ resolves all pairs $u,v$ with $u \in V(G[a_i])$, $v \in V(G[a_j])$.
      
      \item [(c)] Let $u,v \in {\cal V}(c)$. Pair $u,v$ is either resolved by a vertex $w \in R^* \setminus \{ \text{\textnu } (a_i) \} \; | \;  i \in \{1, \ldots , k\}, k \geq 0 \} \subseteq R$ or by a vertex in $\{ \text{\textnu } (a_i) \} \; | \;  i \in \{1, \ldots , k\}, k \geq 0 \}$. Since $R \; \cap \; V(G[a_i]) \neq \emptyset$ for every $i \in \{1, \ldots , k\}, k \geq 0 $ (see Observation \ref{G[a]Ancor}) and every pair that is resolved by \textnu $(a_i)$ is also resolved by a vertex in $R \; \cap \; V(G[a_i])$, $R$ resolves all pairs $u,v \in V(G[a_i])$, see Figure \ref{secondDir3}.    
      
      \item [(d)] Let $u \in V(G[a_i])$ and $v \in {\cal V}(c)$. Since for every $j \in \{1, \ldots , k\}$ vertex \textnu $(a_j)$ is not an $R_j$-gate in $G[a_j]$ or not an $R^*$-gate in $G_{|{\cal V}(c)}$, pair $u,v$ is either resolved by a vertex $w_i \in R_i \setminus \{\text{\textnu } (a_i)\} \neq \emptyset$ (see Observation \ref{G[a]Ancor}) or by a vertex $w \in R^* \setminus \{ \text{\textnu } (a_i) \; | \;  i \in \{1, \ldots , k\}, k \geq 0 \} \neq \emptyset$ (\textnu $(a) \in R^*$), what can be seen as follows:\\
      Assume pair $u,v$ cannot be resolves by a vertex $w_i \in R_i \setminus \{ \text{\textnu } (a_i) \}$ and cannot be resolved by a vertex $w \in R^* \setminus \{ \text{\textnu } (a_i) \; | \;  i \in \{1, \ldots , k\}, k \geq 0\}$. Without loss of generality let \textnu $(a_i)$ be no $R_i$-gate in $G[a_i]$. Then there is a vertex $w_i \in R_i \setminus \{\text{\textnu } (a_i) \}$ such that there is no shortest path from $u$ to $w_i$ via vertex \textnu $(a_i)$.
      Therefore $d_{G[c]}(u,w_i)< d_{G[c]}(u,\text{\textnu } (a_i)) +d_{G[c]}(\text{\textnu } (a_i),w_i)$. Since $w_i$ does not resolve pair $u,v$, we get $d_{G[c]}(v, w_i) = d_{G[c]}(v, \text{\textnu } (a_i)) + d_{G[c]}(\text{\textnu } (a_i),w_i) = d_{G[c]}(u, w_i) < d_{G[c]}(u,\text{\textnu } (a_i)) +d_{G[c]}(\text{\textnu } (a_i),w_i)$. It follows that $d_{G[c]}(v, \text{\textnu } (a_i)) < d_{G[c]}(u, \text{\textnu } (a_i))$.\\
      Vertex $w$ does not resolve pair $u,v$ either, that means $d_{G[c]}(v, w) = d_{G[c]}(v, \text{\textnu } (a_i)) + d_{G[c]}(\text{\textnu } (a_i),w) = d_{G[c]}(u, w) \leq d_{G[c]}(u,\text{\textnu } (a_i)) + d_{G[c]}(\text{\textnu } (a_i),w)$. It follows $d_{G[c]}(v, \text{\textnu } (a_i)) \linebreak \leq d_{G[c]}(u, \text{\textnu } (a_i))$, what contradicts the asumption.
 \end{itemize}

 \begin{figure}[]
 \centering
    \begin{subfigure}[t]{0.45\textwidth}
        \includegraphics[width=\textwidth]{lemma1}
    \caption{
      \label{secondDir1} The Figure shows graph $G[c]$ with subgraph $G[a_i]$ and $G[a_j]$ and vertex pair $u,v \in V(G[a_i]$ that needs to be resolved. They are either resolved by a vertex $w_i \in R_i \setminus \{\text{\textnu } (a_i)\} \subseteq R$, or by vertex \textnu $(a_i)$. Since $R \cap {\cal V}(c) \neq \emptyset$ and every pair that is resolved by \textnu $(a_i)$ is also resolved by a vertex in $R \cap {\cal V}(c)$, $R$ resolves $u,v$.}
    \end{subfigure}
    \quad 
     \begin{subfigure}[t]{0.45\textwidth}
        \includegraphics[width=\textwidth]{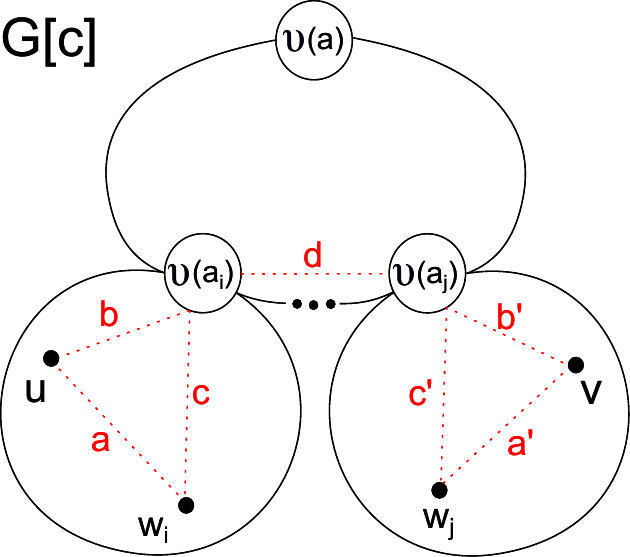}
    \caption{
      \label{secondDir2} The Figure shows graph $G[c]$ with subgraph $G[a_i]$ and $G[a_j]$ and vertex pair $u,v$ with $u \in V(G[a_i])$ and $v \in {\cal V}(c)$ that needs to be resolved. Since $d>0$, either $w_i$ or $w_j$ resolves pair $u,v$.}
    \end{subfigure}  
    \quad
    \begin{subfigure}[t]{0.45\textwidth}
        \includegraphics[width=\textwidth]{lemma2}
    \caption{
      \label{secondDir3} The Figure shows graph $G[c]$ with subgraph $G[a_i]$ and $G[a_j]$ and vertex pair $u,v \in {\cal V}(c)$ that needs to be resolved. They are either resolved by a vertex $w \in R^* \setminus \{ \text{\textnu } (a_i) \} \; | \;  i \in \{1, \ldots , k\}, k \geq 0 \} \subseteq R$ or by a vertex in $\{ \text{\textnu } (a_i) \} \; | \;  i \in \{1, \ldots , k\}, k \geq 0 \}$. Since $R \; \cap \; V(G[a_i]) \neq \emptyset$ for every $i \in \{1, \ldots , k\}, k \geq 0 $ (see Observation \ref{G[a]Ancor}) and every pair that is resolved by \textnu $(a_i)$ is also resolved by a vertex in $R \; \cap \; V(G[a_i])$, $R$ resolves all pairs $u,v \in V(G[a_i])$.}
    \end{subfigure}
    \quad 
     \begin{subfigure}[t]{0.45\textwidth}
        \includegraphics[width=\textwidth]{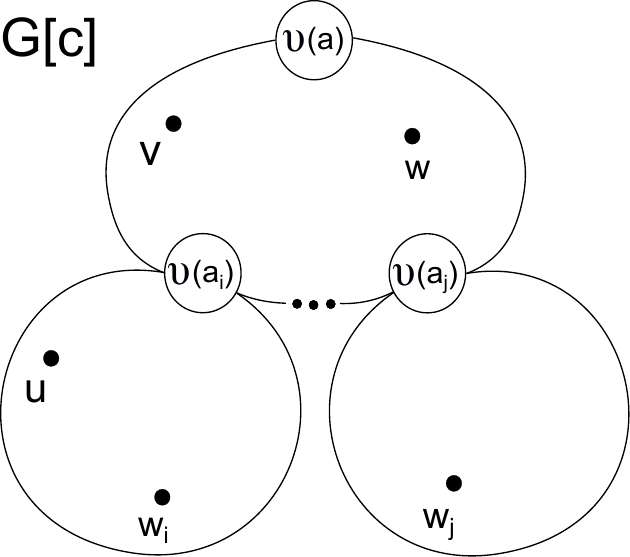}
    \caption{
      \label{secondDir4} The Figure shows graph $G[c]$ with subgraph $G[a_i]$ and $G[a_j]$ and vertex pair $u,v$ with $u \in V(G[a_i])$ and $v \in {\cal V}(c)$ that needs to be resolved. Since for every $i \in \{1, \ldots , k\}$ vertex \textnu $(a_i)$ is no $R_i$-gate in $G[a_i]$ or vertex \textnu $(a_i)$ is no $R^*$-gate in $G_{|{\cal V}(c)}$, pair $u,v$ is either resolved by a vertex $w_i \in R_i \setminus \{\text{\textnu } (a_i)\} \neq \emptyset$ or by a vertex $w \in R^* \setminus \{ \text{\textnu } (a_i) \; | \;  i \in \{1, \ldots , k\}, k \geq 0 \} \neq \emptyset$ (\textnu $(a) \in R^*$).}
    \end{subfigure} 
    \caption{}
\end{figure}
\end{proof}

\begin{proof}{\bf of Theorem \ref{theorem2}}
Graph $G[c]$ is composed by the graph $G_{|{\cal V}(c)}$ and the graphs $G[a_i]$, $i \in \{1, \ldots , k\}$. We compute $h(c) =(\alpha, \beta)$ by computing a minimum-non-gate-\textnu $(a)$-resolving set $A$ for $G[c]$ with $|A| = \alpha$ and a minimum \textnu $(a)$-resolving set $B$ for $G[c]$ with $|B|=\beta$ with the help of Lemma \ref{thLemma}.

Let $A_i$ be a minimum-non-gate-\textnu $(a_i)$-resolving set for $G[a_i]$ and $B_i$ be a minimum-\textnu $(a_i)$-resolving set for $G[a_i]$, $i \in \{1, \ldots , k\}$. To compute sets $A$ and $B$ and thus $\alpha$ and $\beta$ we can do the following:

For every subset $W \subseteq {\cal V}(c)$ that contains vertices \textnu $(a)$, $\text{\textnu }(a_1), \ldots , \text{\textnu }(a_k)$ and resolves all pairs $u,v \in {\cal V}(c)$ we determine a resolving set $R_W$ for $G[c]$. 
$R_W$ contains the vertices in $W$ and for every $i \in \{1, \ldots ,k\}$ either the vertices in $A_i$ or in $B_i$. If vertex \textnu $(a_i)$ is a $W$-gate in $G_{|{\cal V}(c)}$ then $R_W$ contains the vertices in $B_i$ else the vertices in $A_i$. $R_W$ is a \textnu $(a)$-resolving set for $G[c]$ (Lemma \ref{thLemma}) and by Lemma \ref{withoutSep} we get that $R'_W := R_W \setminus \{\text{\textnu }(a_1), \ldots , \text{\textnu }(a_k)\}$ is a  \textnu $(a)$-resolving set for $G[c]$. Vertex set $R'_W$ is a smallest \textnu $(a)$-resolving set for $G[c]$ with the property ${\cal V}(c) \cap R'_W = W$ for a given $W$, that contains at least the vertices of $W$. \\ \null 

Then we have 
$B = \min\{R'_W \; | \; W \subseteq {\cal V}(c) \text{ is a resolving set for } G_{|{\cal V}(c)} \text{ with } \text{\textnu} (a), \text{\textnu }(a_1),\\ \text{\textnu }(a_2), \ldots, \text{\textnu }(a_k) \in W \}$ with $\beta = |B|$ and
$A= \min\{R'_W \; | \; W \subseteq {\cal V}(c) \text{ is a resolving set for  }\\ G_{|{\cal V}(c)} \text{ with } \text{\textnu} (a), \text{\textnu }(a_1), \ldots , \text{\textnu }(a_k) \in W \text{ and \textnu } (a) \text{ is not an $W$-gate in $G[c]$} \}$ with $\alpha = |A|$.
\end{proof}

\begin{theorem}
 The metric dimension of $G[r]$ can efficiently be computed from $h(r)$. 
\end{theorem}

\begin{proof}
 Let $B$ be a minimum-\textnu $(r)$-resolving set for $G[r]$. Since the conditions of Lemma \ref{withoutSep} are given, we get that $B':= B  \setminus \{r\}$ is a resolving set for $G[r]$ and with Lemma \ref{eq} we get that $B'$ is a minimum resolving set for $G[r]$. Thus, if $h(r)=(\alpha,\beta)$, then $\beta-1$ is the metric dimension of $G[r]$.
\end{proof}

\section{Algorithm and Time Complexity} 
\label{AlgTimeComp}

Let $G=(V,E)$ be a connected undirected graph with $|V| = n$ and $|E| = m$. To compute a resolving set for $G$ we first compute the DEBC-tree $\overrightarrow{T} := (V_T,\overrightarrow{E}_T)$ for $G$. This can be done in $\mathcal{O}(n+m)$ with the help of any linear-time-algorithm for finding the biconnected components and bridges of $G$. Then we compute $h(c) = (\alpha_c, \beta_c)$ for every leaf $c$ with father $a$ in the DEBC-tree $\overrightarrow{T}$. We do this by checking every subset $W \subseteq V(G[c])$ if $W' := W \cup \{\text{\textnu }(a)\}$ is a resolving set for $G[c]$. We choose the size of the smallest set $W'$ for $\beta_{c}$ and the size of the smallest set $W'$, such that $a$ is not a $W'$-gate in $G[c]$ for $\alpha$. For this we need $\mathcal{O}(2^n \cdot n \cdot (n+m))$ steps, because we have $2^n$ subset and for each subset we test if it is a resolving set in time $n \cdot (n+m)$. Next we can compute the values $h$ for every inner node of $\overrightarrow{T}$. In Theorem \ref{theorem1} and \ref{theorem2} we showed that this can be done in $\mathcal{O}(1)$ and $\mathcal{O}(2^n \cdot n \cdot (n+m))$, respectively. Overall we have a running time in $\mathcal{O}(2^n \cdot n \cdot (n+m))$.\\ \null

\begin{definition}
 An undirected graph $G$ is {\em (minimum) $k$-EBC-bounded} for some positive integer $k$, if there is a (minimum) resolving set $R$ for $G$ such that every EBC of $G$ contains at most $k$ vertices of $R$. $R$ is called a {\em (minimum) $k$-EBC-bounded-resolving set} for $G$.\\
 Let $\mathcal{G}_k$ and $\mathcal{G}^{min}_k$ be the class of graphs that are $k$-EBC-bounded and minimum-$k$-EBC-bounded, respectively.\\
 A set of graphs $B$ is {\em (minimum) EBC-bounded}, if for every graph $G \in B$ there is a $k$ such that $G$ is (minimum) $k$-EBC-bounded. 
\end{definition}

\begin{corollary}
 The following problems can be solved in polynomial time for any fixed positive integer $k$:
 \begin{enumerate}
  \item Given an undirected graph $G$. Is $G \in \mathcal{G}_k$?
  \item Given a set $B$ of EBC-bounded graphs. Find the smallest integer $k'$ such that $G \in \mathcal{G}_{k'}$.
  \item Given an undirected graph $G \in \mathcal{G}_k$. Compute a minimum $k$-EBC-bounded-resolving set for $G$.
  \item Given an undirected graph $G \in \mathcal{G}^{min}_k$. Compute a minimum resolving set for $G$ and thus the metric dimension of $G$.

 \end{enumerate}
\end{corollary}

To solve these problems we use our algorithm with slight modifications.\\
Instead of checking every subset $W'$ if it is resolving, we do the following:
For the problems 1., 3. and 4. we only test those subsets with at most $k$ vertices. 
For the problem 2. we run our algorithm for $k=1$ and increase $k$ successively by one until we get a resolving set. 
F
By doing these modifications the running time of our algorithm can be bounded by $\mathcal{O}(n^{k} \cdot n \cdot (n+m))$.

Obviously it holds that $\mathcal{G}'_k \subseteq \mathcal{G}_k$. Vice versa it holds that for all $k$ there is a graph $G \in \mathcal{G}_2$ such that $G \notin \mathcal{G}^{min}_k$, see Figure \ref{k_min}. Moreover, the complexity of the following problems remain open:
\begin{enumerate}
 \item Given an undirected graph $G$ a fixed positive integer $k$. Is $G \in \mathcal{G}^{min}_k$?
 \item Given an undirected graph $G \in \mathcal{G}'$. Find the smallest integer $k'$ such that $G \in \mathcal{G}_{k'}$. 
\end{enumerate}

\begin{figure}
 \centering
 \includegraphics[height = 0.3 \textheight]{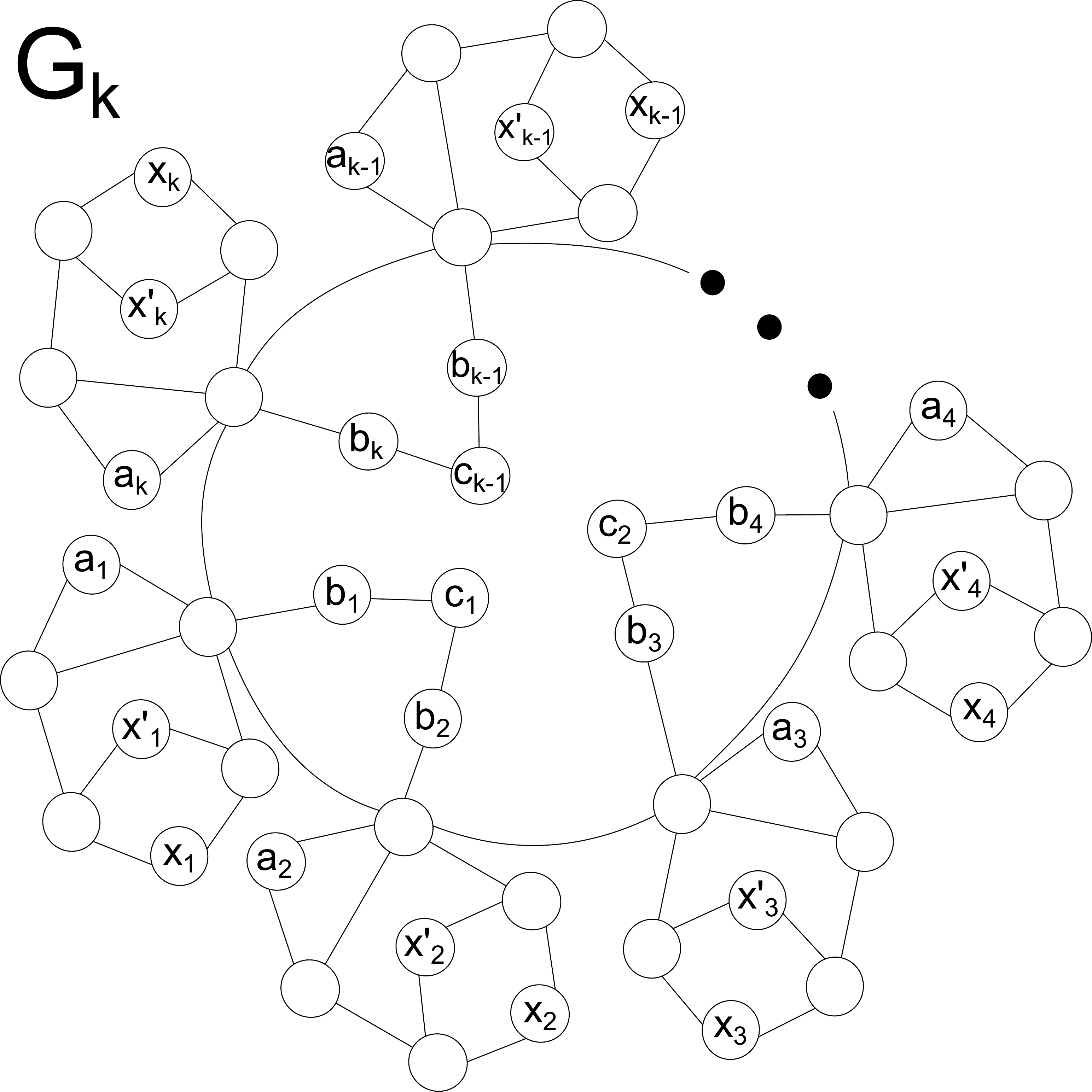}
 \caption{\label{k_min} The Figure shows graph $G_k$ with $k+1$ EBCs. Graph $G_k$ is in $\mathcal{G}_2$ for all $k$. Every resolving set for $G_k$ contains one of the vertices $x_i, x'_i$, $1 \leq i \leq k$, since there is no other vertex that can resolve them. The only vertex pairs that still need to be resolved are pairs $a_i, b_i$. To solve them it suffices to choose vertices $a_i$ as resolving vertices. By doing so we get a 2-EBC-bounded-resolving-set with $2 \cdot k$ vertices (every EBC except of one contains two resolving vertices) and there is no other 2-EBC-bounded-resolving-set with less vertices. Nevertheless a minimum resolving set contains less vertices. By choosing vertices $c_i$ instead of $a_i$ as resolving vertices one gets a minimum resolving set with $\frac{3}{2} \cdot k$ vertices. In this case one of the EBCs contains $\frac{1}{2} \cdot k$ resolving vertices and the others contain one vertex.}
\end{figure}

The following problem, however, still remains NP-complete. The proof can be found in the full version of this paper.

\newpage
\problemDecision{\mdSliceBC}{
  An undirected graph $G=(V,E)$ and a positive integer $r\in\mathbb{N}$ such that there is a minimum resolving set $R\subseteq V$ for $G$ that contains at most $k$ vertices from each biconnected component of $G$.
}{
  Is the metric dimension of $G$ at most $r$?
}  

\begin{theorem}
 {\sc \mdSliceBC} is NP-complete for all positive integers $k\in\mathbb{N}$.
\end{theorem}
  
\begin{proof} 
We use a slight modification of the NP-completeness proof of \sc Metric Dimension \rm from Khuller et al.\ in \cite{KRR96}, where they reduce from \sc 3-SAT. \rm Let $F$ be a \sc 3-SAT \rm instance with $n$ variables and $m$ clauses. For each variable $x_i$ they construct the gadget in Figure \ref{varGad} and for each clause $c_j$ the gadget in Figure \ref{claGad}. The nodes $T_i$ and $F_i$ are the "True" and "False" ends and an variable gadget is connected to the rest of the graph only through these two nodes.

 \begin{figure}[]
 \centering
    \begin{subfigure}[t]{0.3\textwidth}
        \includegraphics[width=\textwidth]{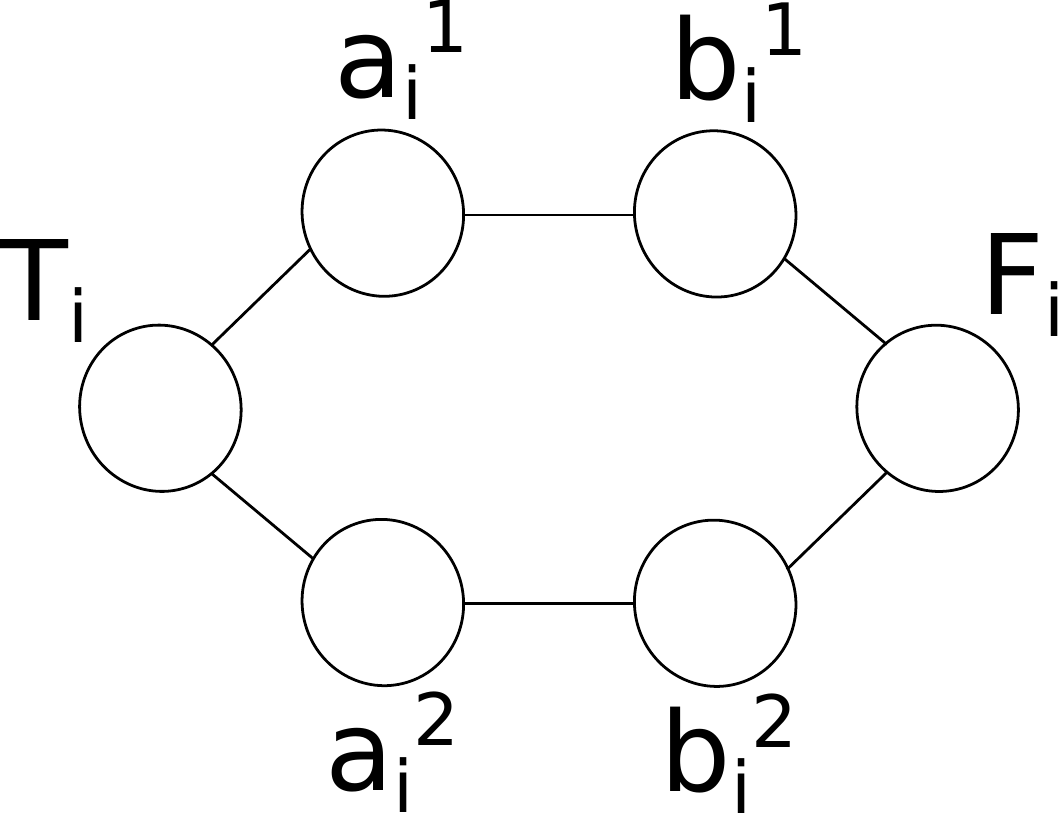}
    \caption{ 
      \label{varGad} Gadget for a variable.}
    \end{subfigure}
    \quad 
     \begin{subfigure}[t]{0.3\textwidth}
        \includegraphics[width=\textwidth]{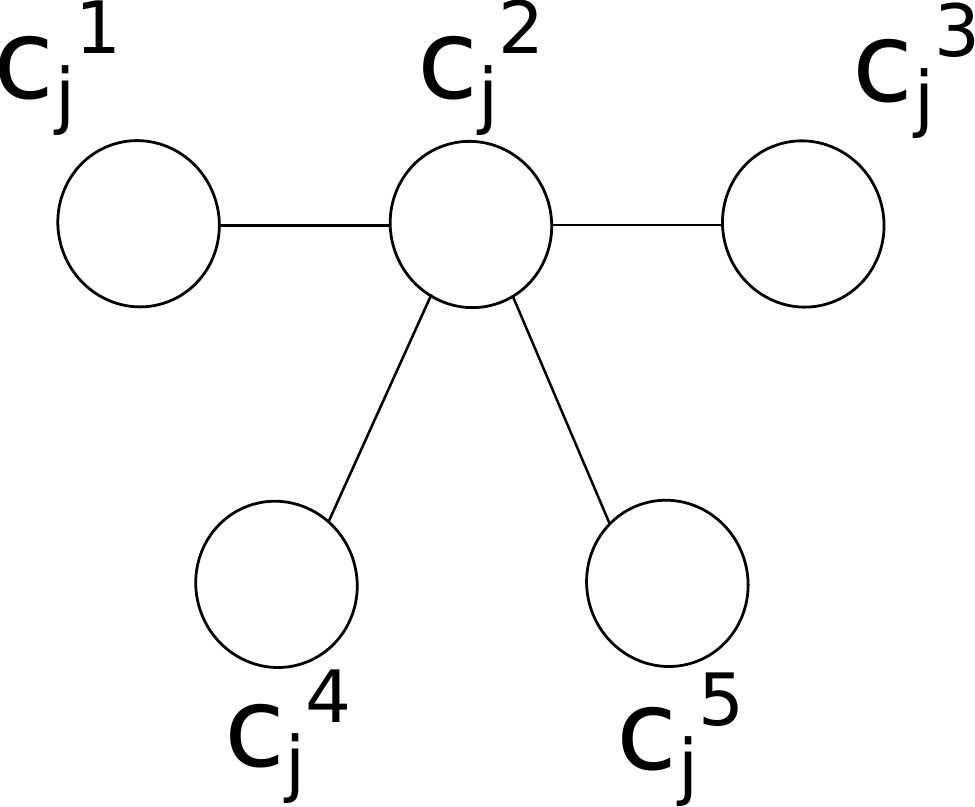}
    \caption{ 
      \label{claGad} Gadget for a clause.}
    \end{subfigure} 
    \quad
     \begin{subfigure}[t]{0.3\textwidth}
        \includegraphics[width=\textwidth]{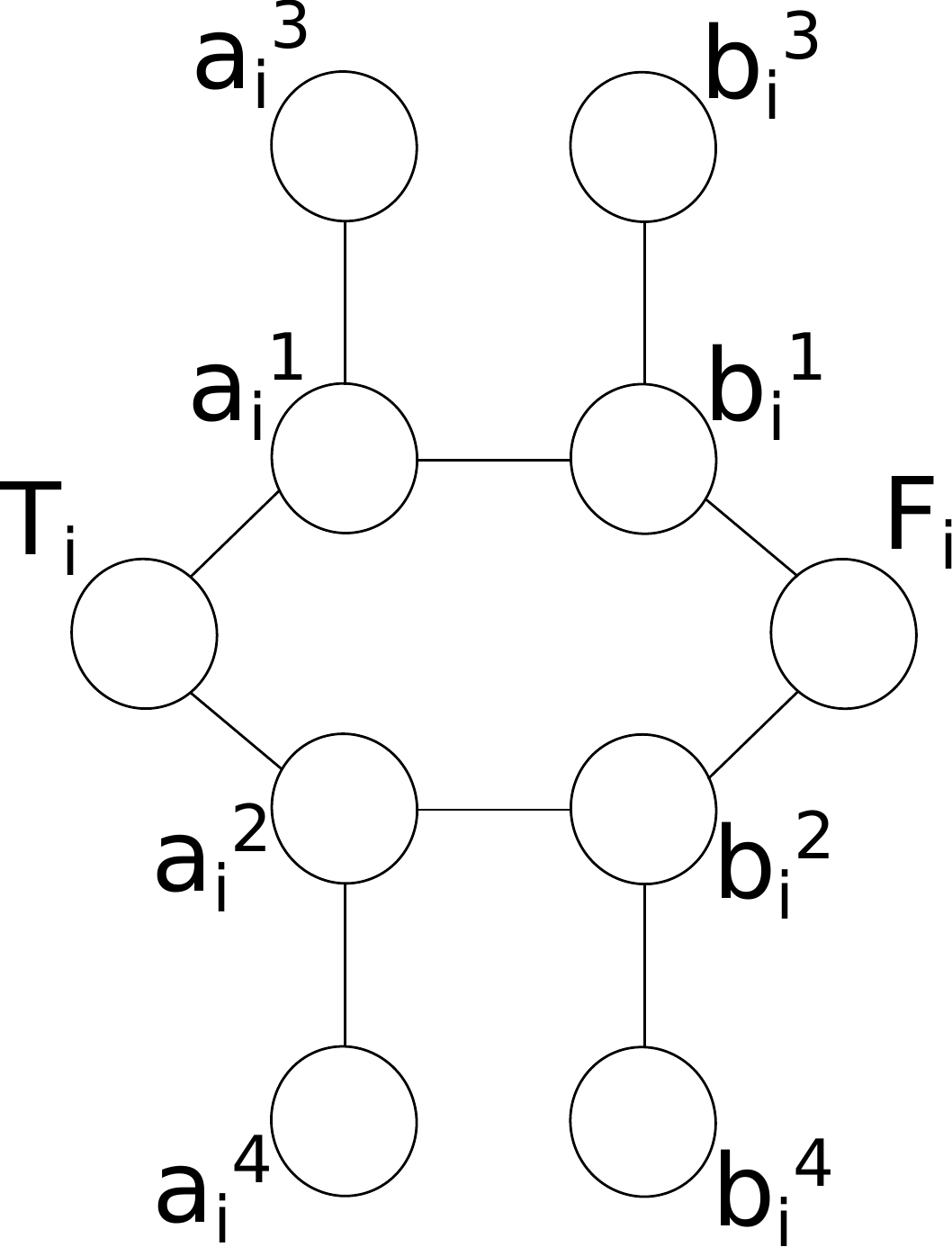}
    \caption{
      \label{varGadMod} Modified variable gadget.}
    \end{subfigure} 
     \caption{}
\end{figure}  

\begin{figure}[] 
\centering
\includegraphics[width = \textwidth]{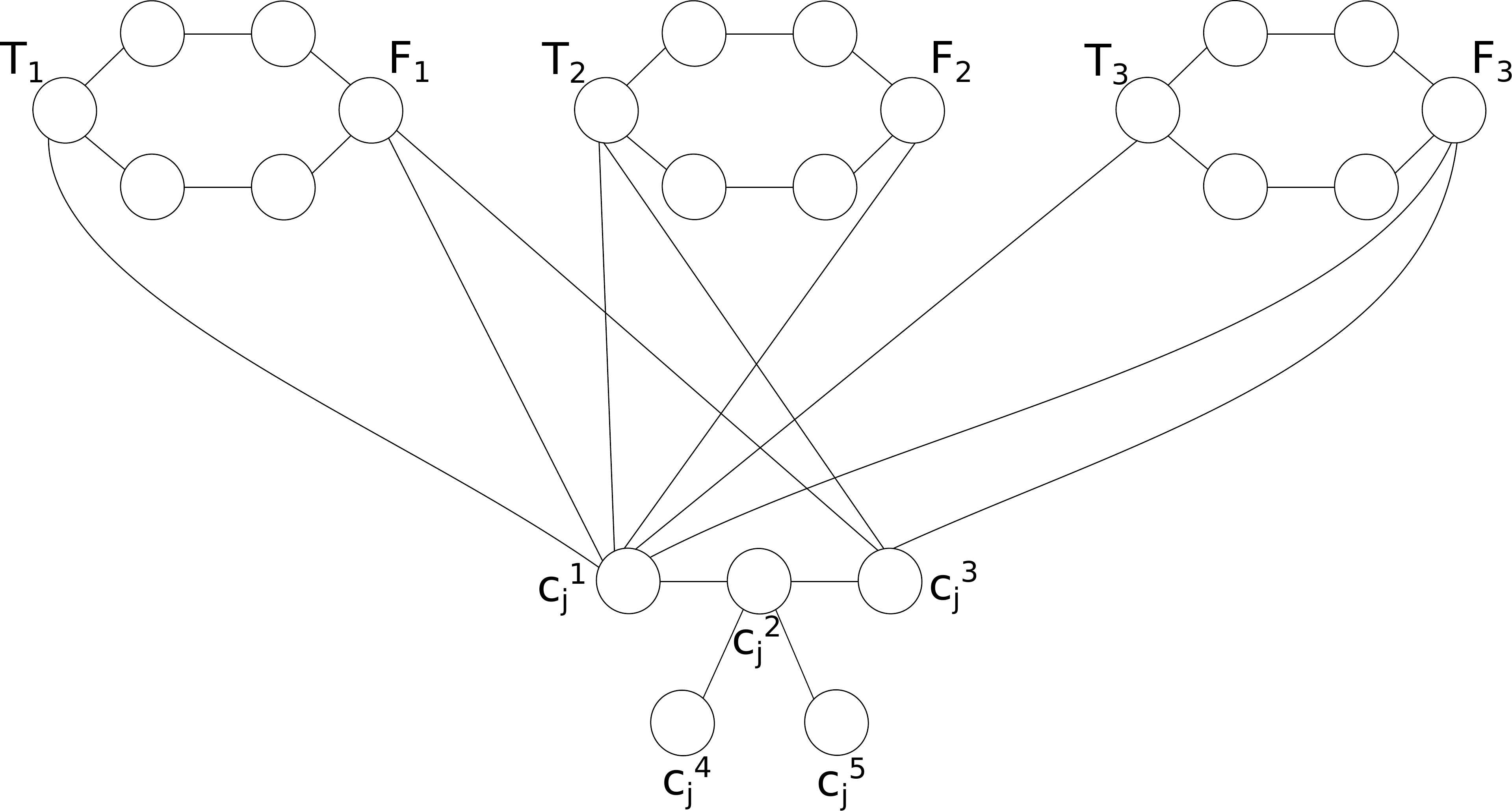}
\caption{\label{variableClauseExample} Clause $C_j = x_1 \vee \bar x_2 \vee x_3$.}
\end{figure}

If $x_i$ is a positive literal in $c_j$ the edges $\{T_i,c_j^1\}$, $\{F_i,c_j^1\}$ and $\{F_i,c_j^3\}$ are added. If $x_i$ is a negative literal in $c_j$ the edges $\{T_i,c_j^1\}$, $\{F_i,c_j^1\}$ and $\{T_i,c_j^3\}$ are added. If $x_i$ does not occur in $c_j$ the edges $\{T_i,c_j^1\}$, $\{F_i,c_j^1\}$, $\{T_i,c_j^3\}$ and $\{F_i,c_j^3\}$ are added, see Figure \ref{variableClauseExample} for an example. In the following they argue that a formula $F$ is satisfiable if and only if the metric dimension of the corresponding graph $G$ is $n+m$. To show this they observe that every resolving set for $G$ contains at least one of the nodes $a_i^1, a_i^2, b_i^1, b_i^2$ and at least one of the nodes $c_j^4, c_j^5$, $1 \leq i \leq n$, $1 \leq j \leq m$.

We now modify the variable gadget as shown in Figure \ref{varGadMod} so that at least one of the nodes $a_i^1, a_i^2, a_i^3, a_i^4,  b_i^1, b_i^2, b_i^3, b_i^4$ must be in every resolving set. Note that whenever one of the nodes $a_i^1, a_i^2, b_i^1, b_i^2$ is in a resolving set it can also be substituted by one of the nodes $a_i^3, a_i^4, b_i^3, b_i^4$.
\end{proof} 

\section{Conclusion}
We have shown that {\sc Metric Dimension} can be solved in polynomial time on graphs having a minimum resolving set with a bounded number of resolving vertices in every EBC. Even more the algorithm even can compute a minimum resolving set in polynomial time under these restrictions. However, the problem remains NP-complete for graphs having a minimum resolving set with a bounded number of vertices in every biconnected component. This shows that the extended biconnected components can not simply be downsized further. A next step can be to investigate how this algorithm can be modified to solve other variants of the Metric Dimension problem, such as the Fault-tolerant Metric Dimension, the Local Metric Dimension, and the Strong Metric Dimension, etc. One version we are currently working on is the Fault-tolerant Metric Dimension as presented in \cite{hernando2008fault}. The two open problems discussed at the end of Section \ref{AlgTimeComp} are also going to be investigated.
 
 \section{References}

 \bibliographystyle{plain} 
 \bibliography{literature}

\begin{thebibliography}{10}

\bibitem{belmonte2017metric}
R{\'e}my Belmonte, Fedor~V Fomin, Petr~A Golovach, and MS~Ramanujan.
\newblock Metric dimension of bounded tree-length graphs.
\newblock {\em SIAM Journal on Discrete Mathematics}, 31(2):1217--1243, 2017.

\bibitem{CGH08}
G.G. Chappell, J.G. Gimbel, and C.~Hartman.
\newblock Bounds on the metric and partition dimensions of a graph.
\newblock {\em Ars Combinatoria}, 88, 2008.

\bibitem{CEJO00}
G.~Chartrand, L.~Eroh, M.A. Johnson, and O.~Oellermann.
\newblock Resolvability in graphs and the metric dimension of a graph.
\newblock {\em Discrete Applied Mathematics}, 105(1-3):99--113, 2000.

\bibitem{CPZ00}
G.~Chartrand, C.~Poisson, and P.~Zhang.
\newblock Resolvability and the upper dimension of graphs.
\newblock {\em Computers and Mathematics with Applications}, 39(12):19--28,
  2000.

\bibitem{DPSL12}
J.~D\'{\i}az, O.~Pottonen, M.J. Serna, and E.J. van Leeuwen.
\newblock On the complexity of metric dimension.
\newblock In L.~Epstein and P.~Ferragina, editors, {\em ESA}, volume 7501 of
  {\em Lecture Notes in Computer Science}, pages 419--430. Springer, 2012.

\bibitem{epstein2015weighted}
Leah Epstein, Asaf Levin, and Gerhard~J Woeginger.
\newblock The (weighted) metric dimension of graphs: hard and easy cases.
\newblock {\em Algorithmica}, 72(4):1130--1171, 2015.

\bibitem{estrada2013k}
Alejandro Estrada-Moreno, Juan~A Rodr{\'\i}guez-Vel{\'a}zquez, and Ismael~G
  Yero.
\newblock The k-metric dimension of a graph.
\newblock {\em arXiv preprint arXiv:1312.6840}, 2013.

\bibitem{fernau2015computing}
Henning Fernau, Pinar Heggernes, Pim van't Hof, Daniel Meister, and Reza Saei.
\newblock Computing the metric dimension for chain graphs.
\newblock {\em Information Processing Letters}, 115(9):671--676, 2015.

\bibitem{foucaud2015algorithms}
Florent Foucaud, George~B Mertzios, Reza Naserasr, Aline Parreau, and Petru
  Valicov.
\newblock Algorithms and complexity for metric dimension and
  location-domination on interval and permutation graphs.
\newblock In {\em International Workshop on Graph-Theoretic Concepts in
  Computer Science}, pages 456--471. Springer, 2015.

\bibitem{foucaud2017identification}
Florent Foucaud, George~B Mertzios, Reza Naserasr, Aline Parreau, and Petru
  Valicov.
\newblock Identification, location--domination and metric dimension on interval
  and permutation graphs. i. bounds.
\newblock {\em Theoretical Computer Science}, 668:43--58, 2017.

\bibitem{GJ79}
M.R. Garey and D.S. Johnson.
\newblock {\em Computers and Intractability: A Guide to the Theory of
  NP-Completeness}.
\newblock W.H.~Freeman, 1979.

\bibitem{HM76}
F.~Harary and R.A. Melter.
\newblock On the metric dimension of a graph.
\newblock {\em Ars Combinatoria}, 2:191--195, 1976.

\bibitem{hartung2013parameterized}
Sepp Hartung and Andr{\'e} Nichterlein.
\newblock On the parameterized and approximation hardness of metric dimension.
\newblock In {\em Computational Complexity (CCC), 2013 IEEE Conference on},
  pages 266--276. IEEE, 2013.

\bibitem{HSV12}
M.~Hauptmann, R.~Schmied, and C.~Viehmann.
\newblock Approximation complexity of metric dimension problem.
\newblock {\em Journal of Discrete Algorithms}, 14:214--222, 2012.

\bibitem{hernando2008fault}
Carmen Hernando, Merc{\'e} Mora, Peter~J Slater, and David~R Wood.
\newblock Fault-tolerant metric dimension of graphs.
\newblock {\em Convexity in discrete structures}, 5:81--85, 2008.

\bibitem{HMPSCP05}
M.C. Hernando, M.~Mora, I.M. Pelayo, C.~Seara, J.~C{\'a}ceres, and M.L.
  Puertas.
\newblock On the metric dimension of some families of graphs.
\newblock {\em Electronic Notes in Discrete Mathematics}, 22:129--133, 2005.

\bibitem{HW12}
S.~Hoffmann and E.~Wanke.
\newblock Metric dimension for gabriel unit disk graphs is {NP}-complete.
\newblock In A.~Bar-Noy and M.M. Halld{\'o}rsson, editors, {\em ALGOSENSORS},
  volume 7718 of {\em Lecture Notes in Computer Science}, pages 90--92.
  Springer, 2012.

\bibitem{hoffmann2016linear}
Stefan Hoffmann, Alina Elterman, and Egon Wanke.
\newblock A linear time algorithm for metric dimension of cactus block graphs.
\newblock {\em Theoretical Computer Science}, 630:43--62, 2016.

\bibitem{IBSS10}
H.~Iswadi, E.~Baskoro, A.N.M. Salman, and R.~Simanjuntak.
\newblock The metric dimension of amalgamation of cycles.
\newblock {\em Far East Journal of Mathematical Sciences (FJMS)}, 41(1):19--31,
  2010.

\bibitem{KRR96}
Samir Khuller, Balaji Raghavachari, and Azriel Rosenfeld.
\newblock Landmarks in graphs.
\newblock {\em Discrete Applied Mathematics}, 70:217--229, 1996.

\bibitem{MT84}
R.A. Melter and I.~Tomescu.
\newblock Metric bases in digital geometry.
\newblock {\em Computer Vision, Graphics, and Image Processing},
  25(1):113--121, 1984.

\bibitem{oellermann2007strong}
Ortrud~R Oellermann and Joel Peters-Fransen.
\newblock The strong metric dimension of graphs and digraphs.
\newblock {\em Discrete Applied Mathematics}, 155(3):356--364, 2007.

\bibitem{SBSSB11}
S.W. Saputro, E.T. Baskoro, A.N.M. Salman, D.~Suprijanto, and A.M. Baca.
\newblock The metric dimension of regular bipartite graphs.
\newblock {\em arXiv/1101.3624}, 2011.

\bibitem{ST04}
A.~Seb{\"o} and E.~Tannier.
\newblock On metric generators of graphs.
\newblock {\em Mathematics of Operations Research}, 29(2):383--393, 2004.

\bibitem{Sla75}
P.~Slater.
\newblock Leaves of trees.
\newblock {\em Congressum Numerantium}, 14:549--559, 1975.

\end{thebibliography}

\end{document}